%% file: main.tex
\documentclass[twoside,leqno]{article}

\usepackage[utf8]{inputenc}
\usepackage[english]{babel}

\usepackage[letterpaper]{geometry}

\usepackage{ltexpprt}
\usepackage{hyperref}

\usepackage{algorithm}
\usepackage{algpseudocode}
\usepackage{amsfonts}
\usepackage{amsmath}
\usepackage{graphicx}

\usepackage{fancyhdr}

\begin{document}

\newcommand\relatedversion{}
\renewcommand\relatedversion{\thanks{This research was supported in part by the Swiss National Science Foundation project 200021--184656 ``Randomness in Problem Instances and Randomized Algorithms'' and the Swiss National Science Fund (SNSF) grant no 200020--182517 “Spatial Coupling and Interpolation Methods for the Analysis of Coding and Estimation”.}} 

\title{\Large An Improved Analysis of Greedy for Online Steiner Forest\relatedversion}
\author{$\acute{\textrm{E}}$tienne Bamas\thanks{$\acute{\textrm{E}}$cole Polytechnique Fédérale de Lausanne.}
\and Marina Drygala\thanks{$\acute{\textrm{E}}$cole Polytechnique Fédérale de Lausanne.}
\and Andreas Maggiori\thanks{$\acute{\textrm{E}}$cole Polytechnique Fédérale de Lausanne.}}

\date{}

\maketitle






\begin{abstract} \small\baselineskip=9pt 
This paper considers the classic Online Steiner Forest problem where one is given a (weighted) graph $G$ and an arbitrary set of $k$ terminal pairs $\{\{s_1,t_1\},\ldots ,\{s_k,t_k\}\}$ that are required to be connected. The goal is to maintain a minimum-weight sub-graph that satisfies all the connectivity requirements as the pairs are revealed one by one. It has been known for a long time that no algorithm (even randomized) can be better than $\Omega(\log(k))$-competitive for this problem. Interestingly, a simple greedy algorithm is already very efficient for this problem. This algorithm can be informally described as follows:

\vspace{2mm}

\textit{Upon arrival of a new pair $\{s_i,t_i\}$, connect $s_i$ and $t_i$ with the shortest path in the current metric, contract the metric along the chosen path and wait for the next pair.}

\vspace{2mm}

Although simple and intuitive, greedy proved itself challenging to analyze and its competitive ratio is a long-standing open problem in the area of online algorithms. The last progress on this problem is due to an elegant analysis by Awerbuch, Azar, and Bartal [SODA~1996], who showed that greedy is $O(\log^2(k))$-competitive.

\vspace{2mm}

In this paper, we identify a natural measure of the ``efficiency" of greedy that we call the \textit{contraction}. The contraction of a pair $\{s_i,t_i\}$ is the ratio between the distance $d_G(s_i,t_i)$ in the graph $G$ and the actual cost that greedy pays for connecting the pair $\{s_i,t_i\}$. Intuitively, a worst-case instance should be an instance on which greedy is very ``inefficient'', i.e. an instance for which all pairs have a relatively small contraction. Indeed, one can remark that all hard instances that appeared in the literature are such that all pairs have a contraction of exactly 1 (which is the smallest contraction possible). Our main result, among others, is to show that greedy is $O(\log(k)\log\log(k))$-competitive on such instances.

At the heart of this new result lies an original use of dual fitting, in which we use the dual solution not only to lower bound the optimum as it is usually the case in competitive analysis, but also to recursively partition the global instance into several disjoint instances of much smaller complexity.

\end{abstract}

\section{Introduction}
\input{Introduction}

\section{Proof of Theorem \ref{thm:main}}
\label{sec:main}
\input{general_theorem}

\section{Proof of Theorem \ref{thm:tree} and Theorem \ref{thm:decreasing}}
\label{sec:variants}
\input{Applications}

\section{Open Problems and Future Directions}
\label{sec:open}
\input{OpenProblems}

\appendix
\section{Deferred proofs}
\label{sec:appendix_proofs}
\input{appendix}

\section*{Acknowledgments}
We are very grateful to our advisor Ola Svensson for insightful discussions on this problem.

\bibliographystyle{siamplain}
\bibliography{main.bib}

\end{document}

%% file: Introduction.tex
One of the most classic problems in network design is arguably the Steiner Tree problem. Given a weighted graph $G=(V,E,w:E\mapsto \mathbb{R}_+)$ and a set $\calT\subseteq V$ of \textit{terminals}, one has to compute the cheapest tree that connects all the terminals. A straightforward generalization of this problem is the so-called Steiner Forest problem in which one is given a set $\calP$ of \textit{pairs} of vertices that are required to be connected. One has to buy the cheapest forest that connects all the pairs of terminals. In a seminal paper, Imase and Waxman \cite{imase1991dynamic} introduced the online version of the Steiner Tree problem. In this version, the terminals in $\calT$ are revealed one by one (in a possibly adversarial order), and the algorithm has to connect all previously arrived terminals before seeing the next one. The challenge lies in the fact that the algorithm is not allowed to remove edges that were bought before. Imase and Waxman provided tight bounds in this scenario by proving that (1) the natural greedy algorithm which simply connects the latest arrived terminal to the closest previously arrived terminal is  $O(\log(|\calT|))$-competitive and (2) no algorithm can be better than $\Omega (\log(|\calT|))$-competitive. Shortly after, Westbrook and Yan \cite{westbrook1995performance} introduced the online Steiner Forest problem. In this variation, the pairs in $\calP$ are revealed one after another (again in a possibly adversarial order). At any point in time, the algorithm has to maintain a feasible solution to the instance made of pairs that arrived so far without discarding previously selected edges. As it is a more general problem, the negative result of Imase and Waxman shows that no algorithm can be better than $\Omega( \log(|\calP|))$-competitive. From now on, with a slight abuse of notation, we will use $k$ to denote either the number of terminals $|\calT|$ in the online Steiner Tree problem or the number of pairs $|\calP|$ in the online Steiner Forest problem. A natural generalization of the greedy algorithm of Imase and Waxman for online Steiner Tree to the case of Steiner Forest can be described informally as follows:

\vspace{2mm}

\textit{Upon the arrival of a new pair $\{s_i,t_i\}$, connect $s_i$ and $t_i$ with the shortest path in the current metric, contract the metric along the chosen path and wait for the next pair.}

\vspace{2mm}

We mention that there are some subtleties about how the metric is contracted exactly, but for the sake of clarity, we will postpone these details to later in the introduction. The reader might think for now that greedy contracts the edges that it selects (i.e. for any edge $e$ selected by greedy, its weight $w(e)$ is set to 0). For now, an algorithm will be considered as ``greedy'' if it always buys the shortest path in the current metric and nothing else. Westbrook and Yan \cite{westbrook1995performance} showed that a wide class of greedy algorithms are $O\left(\sqrt{k}\log(k) \right)$-competitive. This bound was quickly improved by Awerbuch, Azar, and Bartal \cite{SODA96} who showed with an elegant dual fitting argument that greedy algorithms are in fact $O\left(\log^2(k) \right)$-competitive and conjectured that the right bound should be $O(\log(k))$. Since then, their conjecture has remained open. To the best of our knowledge, all lower bounds for greedy that appeared so far in the literature (see \cite{imase1991dynamic,alon1992line,SODA96,chen2010designing}) are instances where the underlying optimum forest is a single tree. Surprisingly, even in this case, nothing better than the $O\left(\log^2(k) \right)$  upper bound is known. We note that the $O(\log(k))$-competitive analysis of Imase and Waxman for online Steiner Tree does not extend to the online Steiner Forest problem, even if we additionally assume that the offline optimum is a single tree. Indeed, the solution constructed by greedy may not be a single connected component even if the offline optimum is. This problem is highlighted by the lower bounds that appear in \cite{SODA96,chen2010designing}, which show the limitations of the current analysis techniques. We defer a more detailed discussion of this phenomenon to Appendix \ref{subsec:appendixexample}.

\paragraph{Previous work on Steiner Forest.} Apart from the two results \cite{SODA96,westbrook1995performance} mentioned above, many papers are related to this problem. The competitive ratio of greedy was first mentioned again in a list of open problems by Fiat and Woeginger \cite{fiat1998online}. Around the same time, Berman and Coulston \cite{berman1997line} designed a more complex (non-greedy) algorithm which they showed to be $O(\log(k))$-competitive. In a nutshell, their algorithm constructs a dual solution in an online manner and uses this dual solution to guide the algorithm on which edges to buy. However, their algorithm is not greedy because it can buy additional edges that could be helpful in the future but are useless right now. Later, Chen, Roughgarden, and Valiant \cite{chen2010designing} applied the result of Awerbuch et al. regarding the greedy algorithm to design network protocols for good equilibrium.
Interestingly, they mention that the non-greedy algorithm of Berman and Coulston was not possible to apply in their setting. Umboh \cite{umboh2014online} introduced a framework based upon tree embeddings to obtain new results in online network design but also to obtain new proofs of known results. For instance, he obtains new proofs that greedy is $O(\log(k))$-competitive in the case of online Steiner Tree and that the Berman-Coulston algorithm is $O(\log(k))$-competitive in the case of online Steiner Forest. However, his framework does not imply any improved bound for greedy algorithms in online Steiner Forest. More recently, the performance of greedy algorithms for online Steiner Forest was further raised as an ``important open problem''  in \cite{dehghani2018greedy} and also cited in \cite{PanigrahiLecture}. 

\paragraph{Further related work.} Online Steiner Tree (or Forest) problems have also attracted attention in various special cases such as outerplanar graphs \cite{matsubayashi2021non}, euclidean metrics \cite{angelopoulos2009competitiveness} or more general cases such as Steiner Tree in directed graphs \cite{angelopoulos2007improved,angelopoulos2008near, faloutsos2002effect}, node-weighted instances \cite{naor2011online,hajiaghayi2013online}, prize collecting or degree-bounded variants \cite{qian2011logn,hajiaghayi2014near,dehghani2016online}. The fact that the competitive ratio of the greedy algorithm is still an open question is reminiscent of a similar situation for the offline Steiner Tree/Forest problems. In the offline case, it was known for a very long time that greedy (i.e. compute a minimum spanning tree of the metric completion on terminals) gives a constant factor approximation to the optimum Steiner Tree. It has only been recently proved by Gupta and Kumar \cite{gupta2015greedy} that a greedy-like algorithm also yields a constant factor approximation in the case of offline Steiner Forest.

\subsection{Our results}

As announced in the introduction, there are some subtleties about how the metric is contracted when running greedy. Hence we will first define three variants of greedy that contract the metric in slightly different ways. However, we emphasize that the best current upper bound for all these three variants is the $O(\log^2(k))$ upper bound of Awerbuch et al. Furthermore, these three variants behave \textit{exactly} the same on the problematic examples of \cite{SODA96,chen2010designing}. In particular, all the discussion so far applies to any of these three variants. Our main theorem will apply to all three variants of greedy, but we will be able to obtain more specialized results for some variants. We believe these specialized results will further demonstrate the interest of our main theorem. Before getting into the precise definition, it is worthwhile to mention that Gupta and Kumar \cite{gupta2015greedy} also discussed some subtleties about contracting the metric (in the offline case) and also defined several algorithms based on that. Hence it is not the first time that altering the contraction procedure has been considered.

\paragraph{Definition of the greedy algorithm.} After greedy connected the latest arrived pair $p=\{s,t\}$, it is clear that the distance between $s$ and $t$ can be set to $0$ in the current metric. This is the only property of the metric contraction that the proof of Awerbuch et al. uses to obtain a good upper bound. As long as we obtain a new graph $G'$ such that $d_{G'}(s,t)=0$, the $O(\log^2 (k))$ upper bound applies. In our paper, each greedy algorithm will formally maintain a graph $G^{(\tau)}=(V,E^{(\tau)})$ that accurately describes the current metric. In any case, greedy will always take the \textit{shortest path} in the current metric $G^{(\tau)}$ to connect a newly arrived pair. After connecting the $\tau$-th pair, the set of edges will be defined as $E^{(\tau)}=E\cup S^{(\tau)}$ where $S^{(\tau)}$ is a set of edges of weight $0$ over vertex set $V$ that we will call \textit{shortcuts}. It will be clear from definition that we will have the natural condition that $\emptyset = S^{(0)}\subseteq S^{(1)}\subseteq \ldots \subseteq S^{(k)}$. The cost incurred by greedy will always be the sum of lengths of all the shortest paths taken for connecting the pairs. We now proceed to define the three contraction rules formally.\\
\vspace{0.6mm}

\noindent\fbox{%
    \parbox{\textwidth}{%
        \textbf{Rule 1:} When greedy connects a pair $\{s,t\}$ through a path $s=v_0,v_1,v_2,\ldots ,v_\ell,t=v_{\ell+1}$, add the following shortcuts:
        \begin{itemize}
            \item For all $0\leq i\leq \ell$ add the edge $\{v_i,v_{i+1}\}$ of weight $0$.
        \end{itemize}
    }%
}
\vspace{1mm}\\
Rule 1 is what Awerbuch et al. intended in their original paper. It can be seen as simply contracting all the edges on the path taken.\\
\vspace{0.6mm}

\noindent\fbox{%
    \parbox{\textwidth}{%
        \textbf{Rule 2:} When greedy connects a pair $\{s,t\}$ add an edge $\{s,t\}$ with weight 0.
    }%
}
\vspace{1mm}\\
Rule 2 is actually how the metric is contracted in \cite{gupta2015greedy} for their main algorithm. Rule 2 might seem much weaker than Rule 1 as fewer shortcuts are added. One can see that for any $u,v\in V$, the distance between $u$ and $v$ can only be smaller when using Rule 1 over Rule 2. Hence the cost of greedy equipped with Rule 2 is always an upper bound on the cost of greedy equipped with Rule 1. However, the proof of \cite{SODA96} already applies in this case. Hence the greedy algorithm that uses Rule 2 is already $O(\log^2(k))$-competitive.\\
\vspace{0.6mm}

\noindent\fbox{%
    \parbox{\textwidth}{%
        \textbf{Rule 3:} When greedy connects a pair $\{s,t\}$ through a path $s,v_1,v_2,\ldots ,v_\ell,t$, let $s=v'_0,v'_1,\ldots ,v'_{\ell'},t=v'_{\ell'+1}$ be the sub-sequence of $s,v_1,v_2,\ldots ,v_\ell,t$ in which we keep only the vertices that appeared in a previous pair (i.e. previously arrived terminals). Then, add the following shortcuts:
        \begin{itemize}
            \item For all $0\leq i\leq \ell'$ add the edge $\{v'_i,v'_{i+1}\}$ of weight $0$.
        \end{itemize}
    }%
}
\vspace{1mm}\\
Intuitively, Rule 3 is in-between rules 1 and 2. It is also reminiscent of the contraction rule of the second algorithm in \cite{gupta2015greedy}. Again, using this rule, it is clear that we obtain shorter paths than with Rule 2. The $O(\log^2(k))$ upper bound also applies when using Rule 3. The formal definition of \greedy follows naturally for any $i\in \{1,2,3\}$.
\begin{algorithm}
\caption{\greedy}\label{alg:greedy}
\begin{algorithmic}[1]
\State Upon arrival of pair $\{s,t\}$, buy the shortest path in the current metric.
\State Update the metric from $G^{(\tau)}$ to $G^{(\tau+1)}$ using Rule $i$ and wait for the next pair.
\end{algorithmic}
\end{algorithm}

As a shorthand, we will denote by $\alg$ the greedy algorithm at hand. If we do not specify which contraction rule we use, it will implicitly mean that the statement that follows holds for any of our three rules.
\vspace{0.9mm}

\paragraph{Our main result.} We can now present our main result and discuss some of its consequences. We introduce an intuitive measure of the efficiency of \greedy. In general, it might be that the cost incurred by the pair $p=\{s,t\}$ is much smaller than $d_{G}(s,t)=d_{G^{(0)}}(s,t)$. Indeed, because of additional shortcuts, it might be that the ratio $d_{G}(s,t)/d_{G^{(\tau)}}(s,t)$ is unbounded when $d_{G^{(\tau)}}(s,t)=0$. We will define this ratio as the \textit{contraction} of pair $p$ and denote it $\alpha(p)$. We have that 
$$1\leq \alpha (p)\leq \infty,$$ for all pairs $p$, in any instance $\I$. Intuitively, a very high contraction means that \greedy did a good job at reusing edges bought before. Following this remark, one can note that all known lower bounds in \cite{imase1991dynamic,alon1992line,SODA96,chen2010designing} have a contraction of exactly 1 for all pairs in the instance (in the case of Steiner Tree instances in \cite{imase1991dynamic,alon1992line}, one can always choose one of the two endpoints of each pair so that this is the case). This seems to confirm the intuitive reasoning that a hard instance should have most of the pairs with low contraction.

After this remark, we use the shorthand $\alg=$ \greedy to denote the greedy algorithm at hand (using any of our three contraction rules). We denote by $\cost_\alg(\I,\calP_{<\alpha})$ the cost incurred by $\alg$ because of pairs of contraction strictly less than $\alpha$ (i.e. pairs $p$ with $\alpha(p)<\alpha$) when running instance $\I$. Note that we do not count the cost incurred by $\alg$ because of pairs with contraction higher than $\alpha$. Furthermore, we will denote $\cost_\alg(\I)$ the total cost incurred by $\alg$. Finally, denote by $w(\OPT(\I))$ the cost of the offline optimum. Our main result is the following.

\begin{theorem}[Main theorem]
\label{thm:main}
Fix a sequence $(\alpha_k)_{k\geq 0}$ with $\alpha_k\geq 1$ for all $k$. Let $\alg$ be a greedy algorithm running on instance $\I$. Then,
$$
    \cost_\alg \left(\I, \calP_{<\alpha_k}\right)=O(\log(k))\cdot  (\log(\alpha_k)+\log\log(k))\cdot w(\OPT(\I)).
$$
\end{theorem}

As mentioned, this theorem applies to greedy with any of our three contraction rules. This already implies the theorem of Awerbuch et al. as it is straightforward to see that pairs with contraction at least $k$ can make greedy pay at most $O(1)\cdot w(\OPT(\I))$. To see this, simply denote by $\calP'$ these pairs with contraction more than $k$ and $c$ the greedy cost of the most expensive pair in $\calP'$. Then greedy pays for those pairs in $\calP'$ at most $k\cdot c$ while $\OPT(\I)$ must pay at least $k\cdot c$ to connect the most expensive pair in $\calP'$. With this observation and plugging in $\alpha_k=k$ in our bound, we obtain an upper bound of $O(\log^2(k))$ on the competitive ratio.

A consequence of our result is that if one wants to have an $\Omega  \left(\log^{1+\epsilon}(k) \right)$ lower bound for some small $\epsilon>0$, it must be that the lower bound on the cost incurred by greedy comes from pairs with contraction at least $2^{\Omega(\log^{\epsilon}(k))}$. This already changes the perspective on how to obtain a stronger lower bound (if it exists) and shows that all previous lower bounds (that have contraction $\alpha_k=1$ for all pairs) cannot give anything stronger than $\Omega (\log(k)\cdot \log\log(k))$. The fact that the cost should come from pairs with high contraction seems quite counter-intuitive, and we believe this is strong evidence that the old conjecture of Awerbuch et al. should be true. Unfortunately, it is not clear to us how to formalize such an intuition. However, this result still has a number of additional consequences that are interesting. In the following, $w(T^\star(\I))$ denotes the optimum \textit{tree} solution to instance $\I$ (i.e. we restrict the solution to be a single connected component).

\begin{theorem}
\label{thm:tree}
    Let $\alg$ be the greedy algorithm using Rule 3. Then,
    $$
        \cost_\alg(\I)= O(\log(k)\log\log(k))\cdot w(T^\star(\I)).
    $$
\end{theorem}

As mentioned in the introduction, even in the case of a single tree spanning the whole graph, nothing better than the general $O(\log^2(k))$ was known, and all lower bounds in the literature are instances where the optimum is a single component. Theorem \ref{thm:tree} also makes us hopeful for another reason. Gupta and Kumar essentially showed that in the offline case, one can assume that the optimum is a single tree (at the cost of losing a constant factor). Then they showed that under this condition, their algorithm is a constant factor approximation. However, extending this proof idea to the online case does not seem straightforward. By using techniques from \cite{gupta2015greedy} in combination with ours, we can obtain the following last result.

\begin{theorem}
\label{thm:decreasing}
    Let $\alg$ be the greedy algorithm using contraction rule 3. Denote by $c_i$ the cost paid by $\alg$ when connecting pair $p_i=\{s_i,t_i\}$ and define $d_i=d_{G}(s_i,t_i)$ (note that we take the distance in the original graph $G$). If either of the two sequences $(c_1,c_2,\ldots ,c_k)$ or $(d_1,d_2,\ldots ,d_k)$ is non-increasing then 
    $$
        \cost_\alg(\I)= O(\log(k)\log\log(k))\cdot w(\OPT(\I)).
    $$
\end{theorem}

This last result can be derived by using our main theorem, combined with the same potential function argument of \cite{gupta2015greedy}. The assumption is that either the costs paid by $\alg$ or the shortest paths in the original metric are non-increasing over time. This last result is interesting for the following reason: All lower bounds in the literature use the same strategy to fool greedy. First ask two terminals that are very far apart, greedy connects them through a path and then the endpoints of the second pair are closer but not on the path that greedy bought. Then one can repeat this process with the endpoints of the pairs that are getting closer and closer. This forces greedy to buy another path for every new pair. Crafting a difficult instance not following this behavior seems very difficult, especially in the light of the result of \cite{gupta2015greedy}. They show that their greedy-like algorithm, that connects the closest terminals which are not yet satisfied (i.e. the opposite behavior of the assumption of Theorem \ref{thm:decreasing}) is a constant factor approximation for the offline Steiner Forest.
\vspace{0.5mm}
\paragraph{} The takeaway message of all our results combined is that an $\omega (\log(k)\cdot \log\log(k))$ lower bound cannot be any of the following:
\begin{itemize}
    \item An instance in which all pairs have a small contraction, or
    \item an instance where a tree is a good solution, or
    \item an instance with non-increasing costs or shortest paths.
\end{itemize}
\vspace{0.5mm}
However, all the known examples providing an $\Omega (\log(k))$ lower bounds satisfy these conditions. We believe our results change the perspective on this problem, and it would be surprising if an $\Omega(\log^{1+\epsilon}(k))$ lower bound exists for some $\epsilon>0$.

\paragraph{} The rest of the paper is organized as follows. In Subsection \ref{subsec:techoverview}, we present an overview of our techniques that lead to new results. In Section \ref{sec:main} we present the proof of our main result, Theorem \ref{thm:main}. In Section \ref{sec:variants}, we present our proof of Theorem \ref{thm:tree} and Theorem \ref{thm:decreasing}. Finally in Section \ref{sec:open} we discuss the questions left open in this work.

\subsection{Our Techniques}
\label{subsec:techoverview}

In this subsection, we give an overview of the techniques to prove Theorem \ref{thm:main}. We will give the necessary technical details but keeping it rather informal. The complete and formal proof of Theorem \ref{thm:main} appears in Section \ref{sec:main}.
\vspace{0.5mm}
\paragraph{Previous techniques.} Before going into our techniques, we will briefly mention the main ingredients of previous proofs. For ease of notation, we will denote by $\calT$ the set of terminals that appear in at least one pair of $\calP$. For a terminal $u \in \calT$, denote  by $\overline{u}$ its \textit{mate} which is the terminal that $u$ should be connected to. Note that we can assume without loss of generality that all terminals have exactly one mate, by duplicating vertices if this is not the case. For ease of presentation, we will assume until the end of the section that for each pair $\{u, \overline{u}\}$ greedy pays exactly the distance of $u$ and $\overline{u}$ in the initial graph, i.e. no previously arrived pair helps greedy in paying less for the newly arrived terminal pair $\{u, \overline{u}\}$. Put otherwise, we assume that the contraction of all the pairs is equal to 1. 

By standard arguments, we can assume that $w(\OPT (\I)) = k$ and that for each pair of terminal greedy pays a cost that belongs to the set $\{ k/ 2^i\}_{1\leq i\leq \log(k)}$, with only the loss of a constant factor. Indeed, we can rescale all the edges in the graph $G$ by the factor $k/w(\OPT(\I))$ and for the second assumption, note that by standard geometric grouping, we can assume that greedy pays a cost that belongs to the set $\{ k/ 2^i\}_{1\leq i\leq \infty}$. It is then straightforward to see that the total cost incurred by greedy for pairs cheaper than $k/2^{\log(k)}=1$ is at most $O(w(\OPT(\I))$. Based on this observation, we will partition $\calP$ into disjoint sets $\calP^{(i)}$ where each set contains pairs of terminals for which greedy paid $k / 2^i$. These sets will be called \textit{cost classes}. Moreover in order to introduce the first observation let $B(v,r)=\{u\in V \mid d_{G}(v,u)<r\}$ be an open ball with center the terminal $u$ and radius $r$. Then a classic observation is the following.

\begin{observation}
\label{obs: disjoint dual balls}
Let $\calB=\{B_1,B_2,\ldots ,B_\ell\}$ be a collection of balls around terminals, such that (1) all the balls are pairwise disjoint and (2) each ball $B_i$ is centered as some terminal $u\in \calT$ and its radius $r_i$ satisfies $r_i<d_{G}(u,\overline{u})$. Then $\sum_i r_i \leq w(\OPT (\I))$.
\end{observation}
The proof of this observation is straightforward, and we will come back to it later in Section \ref{sec:main}. These balls can be viewed as a solution to the dual of the natural linear programming relaxation of the Steiner Forest problem; hence we will refer to these balls as dual balls. We continue by restating (informally) a key lemma in the analysis of \cite{SODA96}.



\begin{lemma}[\cite{SODA96}]
\label{lem:disjoint balls for each class}
For any cost class $\calP^{(i)}$ (associated with cost $c_i=k/2^i$), it is possible to place disjoint dual balls in $G$ such that these balls are centered around terminals that belong to $\calP^{(i)}$ and they all have a radius of 
$$
    \frac{k}{2^i\cdot \log(|\calP^{(i)}|)}.
$$
Moreover, at the cost of losing a constant factor we can assume that every pair $p=\{s,t\}$ in $\calP^{(i)}$ has at least one dual ball centered around $s$ or $t$. 
\end{lemma}

By taking Lemma \ref{lem:disjoint balls for each class} together with Observation \ref{obs: disjoint dual balls}, we obtain that greedy pays at most $O(\log(|\calP^{(i)}|))=O(\log(k))$ times the dual solution for each cost class. This proves that greedy is $O (\log (k) )$-competitive for each cost class $\calP_i$. Hence we see that the previous proof technique has mainly two ingredients:
\begin{enumerate}
    \item[(1)] Partition the set $\calP$ into disjoint cost classes $\calP^{(i)}$ such that $\calP = \bigcup_{i=1}^{\log (k)} {\calP}^{(i)}$ and for each pair in ${\calP}^{(i)}$ greedy paid $k/2^i$.
    \item[(2)] Prove that greedy is $O ( \log (k) )$-competitive for each cost class separately by building a dual solution.
\end{enumerate}
By (1) and (2) we get that greedy is $O (\log (k) \cdot \log (k) ) = O (\log^2 (k))$-competitive. Interestingly, in the case of Steiner Tree, it is possible to improve the second step by showing that greedy is $O(1)$-competitive for each cost class, hence the $O(\log(k))$ competitive ratio in general (see \cite{alon1992line}). Unfortunately, this is impossible in the case of Steiner Forest. Even if the underlying optimum is a spanning tree, it might be that greedy is already $\Omega(\log(k))$-competitive for a single cost class (see Appendix \ref{subsec:appendixexample} for an example). We also note that the Berman-Coulston algorithm \cite{berman1997line} is designed so that the algorithm is $O(1)$-competitive for each cost class, so the analysis of this more complex algorithm cannot apply to greedy.

\paragraph{Our new approach.} The first step of our new proof relies in a different partitioning of the set $\calP$. Indeed we will partition $\calP$ into $N = \Theta(\log \log (k))$ classes $\tilde \calP^{(j)}$ such that each class is defined as follows.
$$
     \tilde \calP^{(j)}=\bigcup_{i \equiv j \Mod{N}}\calP^{(i)}.
$$

Note that we have $\Theta\left(\log\log(k) \right)$ groups with this partitioning, each containing $\Theta\left(\frac{\log(k)}{\log\log(k)} \right)$ cost classes. Inside each group, the cost classes have the nice property that they are \textit{well separated}, that is, the multiplicative gap between two consecutive costs is $\textrm{polylog}(k)$. We will make good use of this property to disentangle the interactions between pairs that have different costs. Using these techniques we prove that the competitive ratio of greedy for each set $\tilde \calP^{(j)}$ is $O (\log (k))$ ending up with a competitive ratio of  $O (\log (k) \cdot  \log \log (k) )$ overall. The main technical challenge lies in proving such a result.

\begin{figure}
\centering
\includegraphics[scale=0.7]{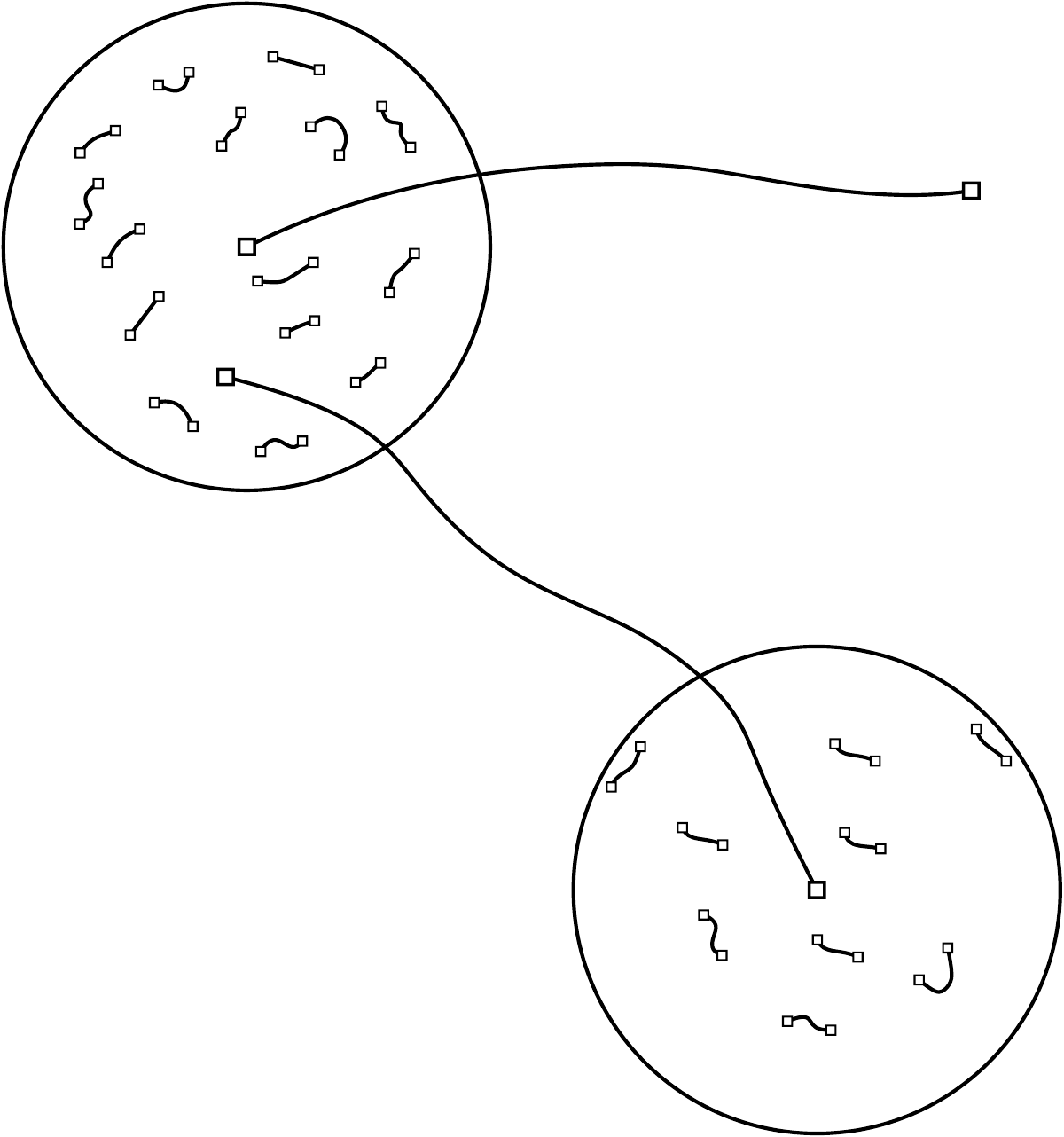}
\caption{An example with two cost classes before charging or clustering}
\label{fig:beforecharge}
\end{figure}
\begin{figure}
    \centering
    \includegraphics[scale=0.7]{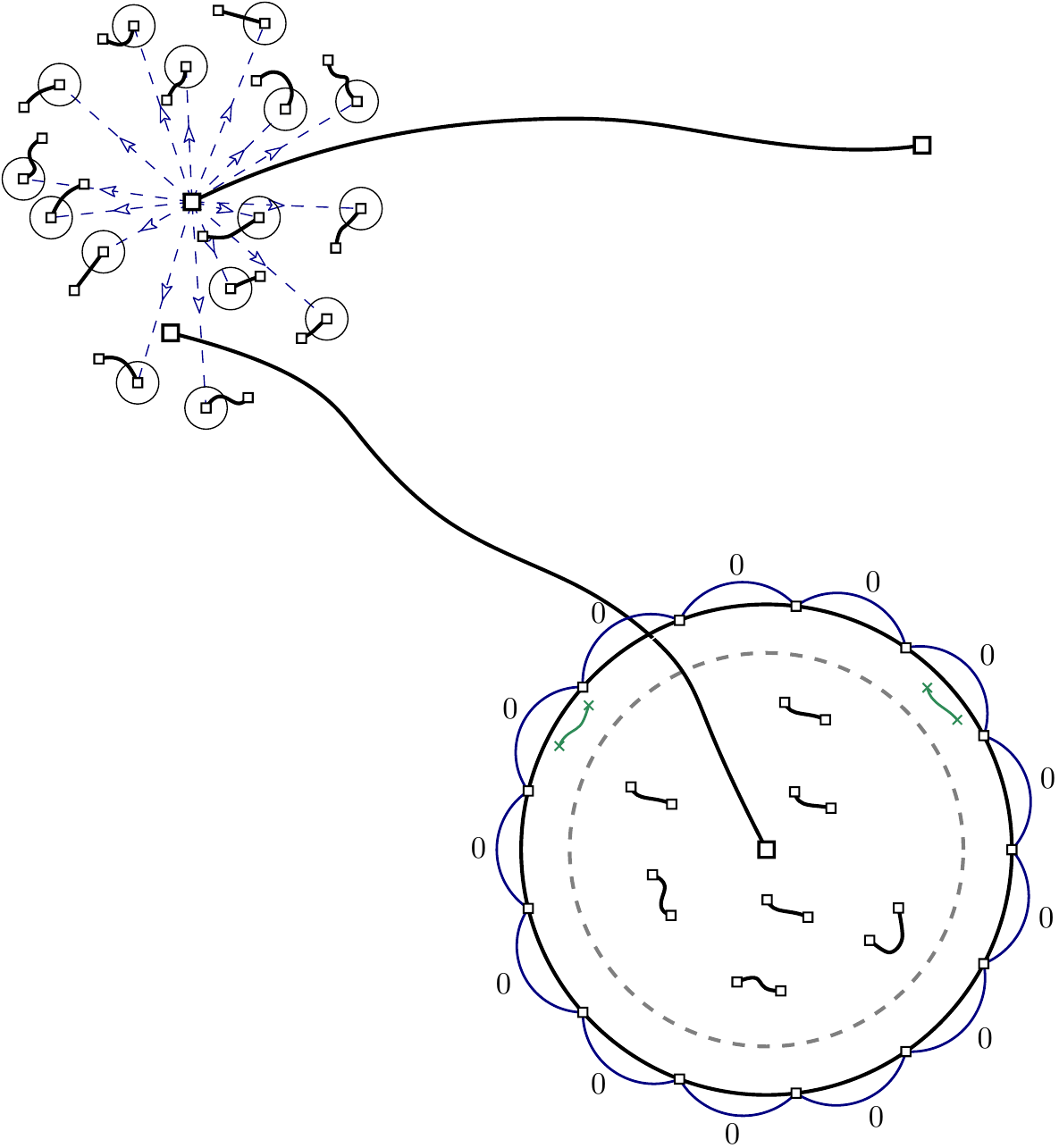}
    \caption{An example with two cost classes after charging or clustering. In the top left corner, smaller pairs are much more expensive than the pairs that created the big dual ball while we have the opposite situation in the bottom right corner.}
    \label{fig:aftercharge}
\end{figure}

If we use Lemma \ref{lem:disjoint balls for each class} to place dual balls around pairs in each cost class $\calP^{(i)}\subseteq \tilde \calP^{(j)}$ (hence creating several collections $\calB^{(i)}$ of balls) it might be that two dual balls $B,B'$ that belong to different sets $\calB^{(i)}$ and $\calB^{(i')}$ overlap. This is the critical issue in the previous proofs, and we will proceed differently. For simplicity, assume we have two cost classes in our set $\tilde \calP^{(j)}$. Let $c_i$ be the cost of the larger class and $c_{i'}$ the cost of the smaller class (hence $i'>i$). We place the dual balls only for the biggest cost class (using Lemma \ref{lem:disjoint balls for each class}). Intuitively, the worst case in the analysis will be when \textit{all} pairs of the smaller cost lie inside the dual balls from the bigger cost class (as depicted in Figure \ref{fig:beforecharge}). If this happens, it will be impossible to place dual balls for the small cost class without intersecting the bigger balls already placed. To overcome this issue, we consider a ball $B$ from the big cost class, and we look at the number of pairs from the small cost class that lie inside this ball. Let the number of such pairs be $k'$. We have two cases.
\begin{enumerate}
    \item[(1)] $k'\cdot c_{i'}\geq \textrm{polylog}(k)\cdot c_i$ which is the easy case. In this case, instead of charging the cost of the big pairs to the dual ball $B$ we can instead charge this cost to the smaller pairs inside $B$. By Lemma \ref{lem:disjoint balls for each class}, the cost that was initially charged to the ball $B$ was $O(c_i)$. Hence if we evenly distribute this cost among all the small pairs inside $B$, each pair will get a cost of roughly.
    $$
        O\left(\frac{c_i}{k'}\right)\leq O\left( \frac{c_{i'}\cdot c_i}{\textrm{polylog}(k)\cdot c_i}\right)=O\left(\frac{c_{i'}}{\textrm{polylog}(k)}\right).
    $$
    Since the cost was transferred to smaller pairs, we can also safely delete the big dual ball $B$, hence making this space available to place the smaller balls. Note that smaller pairs can be charged at most once in this way because the balls in the dual solution $\calB^{(i)}$ for big pairs are pairwise disjoint. This case is depicted in the top left corner of Figure \ref{fig:aftercharge}.
    
    \item[(2)] $k'\cdot c_{i'}<\textrm{polylog}(k)\cdot c_i$. This is the most challenging case and is depicted in the bottom right corner of Figure \ref{fig:aftercharge}. We cannot proceed as in the previous case as we cannot guarantee that small pairs do not get charged too much. Here lies the crux of our proof. First, by re-scaling slightly the ball $B$, we can assume that \textit{almost} \textit{all} the small pairs in $B$ are far from the border of $B$ (a pair $\{s,t\}$ is \textit{far} from the border if one of $s$ or $t$ is at a distance much bigger than $c_{i'}$ from the border of $B$). For simplicity we assume that all the small pairs are far from the border of $B$. From here we construct an instance $\I'$ as follows. Consider the graph $G'$ that is induced by vertices inside $B$. The instance $\I'$ will be composed of the set $\calP'$ containing all the pairs of small cost that are inside $B$, and the metric will be the graph $G'$. Here the assumption that the contraction is 1 implies that both endpoints of each pair in $\calP'$ should be inside $B$ (one endpoint well inside the ball and one outside would cost too much). It also implies that greedy behaves \textit{exactly} the same for the pairs in $\calP'$ in instance $\I'$ as it was behaving for these pairs $\calP'$ in instance $\I$, that is, for each pair $p\in \calP'$ greedy buys exactly the same path to connect $p$ whether instance $\I$ or $\I'$ is running. Recall that we assumed 
    $$
        k'\cdot c_{i'}<\textrm{polylog}(k)\cdot c_i
    $$ hence we have 
    $$
        k'< \frac{\textrm{polylog}(k)\cdot c_i}{c_{i'}}.
    $$
    We know by previous results that greedy is $O(\log(k))$-competitive on a single cost class; hence the competitive ratio of greedy on instance $\I'$ will be bounded by 
    $$
        O(\log(k'))=O\left(\log\left(\frac{\textrm{polylog}(k)\cdot c_i}{c_{i'}}\right)\right)=O(\log\log(k))+O\left(\log\left(\frac{c_i}{c_{i'}}\right)\right).
    $$
    Now the crucial question: What is the value of $w(\OPT(\I'))$? As we defined the graph $G'$ now, it is not clear. But because we assumed all the small pairs are far from the border of $B$, we can allow ourselves to modify the metric on the border of $B$ \textit{without} changing the behavior of greedy. If we consider $V'$ the set of vertices that lie \textit{exactly} on the border of $B$ we will add an edge of length 0 between any pair of vertices in $V'$. This does not change the behavior of greedy on instance $\I'$ because these extra edges are already too far from the pairs in $\I'$ to be used (see Figure \ref{fig:aftercharge}, bottom right corner). The interesting fact is now that 
    $$
        w(\OPT(\I'))\leq w(\OPT(\I)\cap B),
    $$ where we denote by $w(\OPT(\I)\cap B)$ the cost of edges bought by $\OPT(\I)$ inside the ball $B$.
\end{enumerate}

These observations suggest a Top-Down approach where we first try to place dual balls around big pairs. Then proceed by the case distinction described above. Then we move to the next cost class but ignoring all the pairs that got into case (2). We repeat this until we reach the bottom of the cost hierarchy. We end up with dual balls that have different radii but are all pairwise disjoint (because we ignored the pairs that were in case (2) of any iteration). During this process, each pair got into case (1) at most $\log(k)$ times, hence the total additional charge is $O(\frac{\log(k)}{\textrm{polylog}(k)})=O(1)$. It remains to handle all the pairs that were ignored. The idea is now that these pairs can be partitioned into disjoint instances with not too many pairs (recall that we have an upper bound on $k'$ in case (2)) and such that the optimum solution is at most what $\OPT(\I)$ pays inside the ball that created this instance. For instance in the case of two consecutive cost classes $\calP^{(i)},\calP^{(i')}$ (hence $\frac{c_i}{c_{i'}}=\textrm{polylog}(k)$), the total cost of ignored pairs would be:

$$
    \sum_{B\in \calB^{(i)}} \left(O(\log\log(k))+O\left(\log\left(\frac{c_i}{c_{i'}}\right)\right) \right)\cdot w(\OPT(\I)\cap B) = \sum_{B\in \calB^{(i)}}O(\log\log(k))\cdot  w(\OPT(\I)\cap B).
$$

But because the pairs in $\calB^{(i)}$ are pairwise disjoint we have $\sum_{B\in \calB^{(i)}}w(\OPT(\I)\cap B)\leq w(\OPT(\I))$. Hence in total the ignored pairs cost at most $O(\log\log(k))\cdot w(\OPT(\I))$ to greedy. Because we have $\Theta (\frac{\log(k)}{\log\log(k)})$ cost classes inside a set $\tilde \calP^{(j)}$, it feels that a $\log(k)$ competitive ratio for this set is now possible. Of course we took two consecutive cost classes so that $\frac{c_i}{c_{i'}}=\textrm{polylog}(k)$ but this is intuitively the worst case in the analysis. All this is formally handled via a delicate induction that is done in Section \ref{sec:main}.

%% file: general_theorem.tex
This section is devoted to the proof of Theorem \ref{thm:main}. Recall that this theorem applies to the three variants of greedy as defined in the introduction. Hence in this section, $\alg$ will denote \greedy for any $i\in \{1,2,3\}$. This section is organized as follows. In Subsection \ref{subsec:maindef}, we introduce some basic definitions that will be needed. In Subsection \ref{subsec:mainprelim}, we detail some results implied by previous work as well as some pre-processing of the instance needed for the rest of the proof. Namely, we recall the concept of dual fitting used by \cite{SODA96} . In addition, we pre-process the instance so that the different costs greedy pays upon the arrival of different pairs is well-structured (i.e. there is a geometric grouping and a big gap in-between two consecutive cost classes). In Subsection \ref{subsec:mainoverview}, we give an overview of the main body of the proof, and finally, in Subsections \ref{subsec:mainbalanced} and \ref{subsec:maininductive}, we finish the proof.

\subsection{Problem definition and notation}
\label{subsec:maindef}
We will consider a slightly more general problem than Online Steiner Forest. Formally, we are given a weighted graph $G=(V,E,w)$ with weight function $w:E\mapsto \mathbb{R}_{\geq 0}$. Along with graph $G$ we are given an ordered sequence of pairs of vertices $\calP=\{p_1,p_2,\ldots , p_k \}$ revealed one by one, and an ordered sequence of sets $\Se=\{E_1,E_2,\ldots ,E_k\}$ of additional weighted edges. These edges will be made available to the online algorithm $\alg$ over time as follows. Before revealing the first pair $p_1$, the set of edges in $E_1$ is added in the graph $G$ to form the graph $G_1=(V,E\cup E_1)$. These edges in $E_1$ will remain available to the greedy algorithm $\alg$ until the end. Then $\alg$ buys some path and contracts the metric according to its contraction rule as defined in the Introduction to obtain graph $G'_1$. Next, before revealing $p_2$, we add the edge set $E_2$ into the graph $G'_1$ to obtain the graph $G'_2$ (hence $\alg$ updates the metric accordingly). Then $\alg$ sees the pair $p_2$ and so on. In general, if $G^{(\tau)}$ denotes the current metric available to greedy after reading pairs $p_1,p_2,\ldots ,p_\tau$, we first add the edges of $E_{\tau+1}$ to the graph $G^{(\tau)}$ and after this $\alg$ connects the pair $p_{\tau+1}$ via the shortest path contracting the metric according to the chosen contraction rule. We call this variant \textit{online Steiner Forest in decreasing metrics}. This generalizes the classic Online Steiner Forest which is the special case where $E_i=\emptyset$ for all $i$.

The goal is to compare the cost incurred by $\alg$ on the instance $\I=(G,\calP,\Se)$ to the cost of the optimum Steiner forest in the graph $G$ with pairs $\calP$. We insist that the offline optimum is not allowed to use edges from $\Se$ while the algorithm $\alg$ can use these edges in $\Se$ after they are revealed to it. We will denote the optimum cost by $w(\OPT(\I))$. The \textit{size} of an instance $\I$ is the number of pairs in $\calP$. It will be denoted $k$ in the following (hence $k=|\calP|$). For each pair $p=\{s,t\}$, we will naturally call the \textit{endpoints} of the pair $p$ the two vertices $s,t$.

\vspace{0.2cm}
For any subset $S\subseteq \calP$, we will denote by $\cost_{\alg}(\I,S)$ the cost incurred by algorithm $\alg$ on instance $\I$ because of pairs in $S$. By a slight abuse of notation, for a single pair $p$, we will denote by $\cost_{\alg}(\I,p)=\cost_{\alg}(\I,\{p\})$ the cost that $\alg$ pays upon arrival of $p$. 

The \textit{contraction} of a pair with respect to instance $\I$ and algorithm $\alg$ will be the ratio of the shortest path distance in-between the two endpoints of the pair in $G$ and the actual cost paid by $\alg$ for this pair when running instance $\I$. Note that the shortest path is taken in the original graph $G$, without help of edges in $\Se$. Formally, if we denote by $\alpha(p)$ the contraction of pair $p=\{s,t\}$, we have
$$
    \alpha(p)=\frac{d_G(s,t)}{\cost_{\alg}(\I,p)},
$$

with the convention that $\alpha(p)=\infty$ if $\cost_{\alg}(\I,p)=0$. Given a fixed $\alpha\geq 1$ and an instance $\I=(G,\calP,\Se)$, we denote by $\calP_{<\alpha}$ the set of pairs of $\calP$ that have contraction less than $\alpha$ when running $\alg$ (i.e. the pairs $p$ with $\alpha(p)<\alpha$).

\vspace{0.2cm}
For simplicity we will assume that every edge is of weight exactly $\eta$ for some arbitrarily small $\eta>0$. If the graph $G$ does not satisfy this, we subdivide all the edges into chains of smaller edges. Of course, this increases the number of edges and vertices in the graph, but since our competitive ratio is only a function of the number of pairs of terminals, this subdivision will not hurt our analysis. It is also clear that subdividing edges changes neither the optimum solution nor the behavior of the greedy algorithm. It also does not change any of the parameters we just defined above. This assumption will be used for simplicity when constructing balls in the graph; we will assume that no edge in the graph $G$ has an endpoint inside and the other endpoint outside the balls (i.e. edges do not "jump over" the border of any ball). 

\subsection{Preliminary results and preprocessing of the instance}
\label{subsec:mainprelim}
We describe here some key concepts that will be useful in the rest of the proof. We first introduce the following definition that gives much more structure to the instance $\I$.

\begin{definition}[Canonical instance]
\label{def:canonical}
For any $\alpha,\delta \geq 1$, any instance $\I=(G,\calP,\Se)$ of online Steiner Forest in decreasing metrics is said to be $(\alpha,\delta)$-\textit{canonical} with respect to a greedy algorithm $\alg$ if the following holds:
\begin{enumerate}
    \item[(a)] There exist some real number $m>0$ such that for any pair $p\in \calP$, there exists an integer $1\leq   j \leq  \log(k)/\delta$ such that 
    $$
        \cost_{\alg}(\I,p) =  m/2^{j\cdot (\delta+10)}.
    $$
    (i.e. we have some geometric grouping of costs and two consecutive cost classes are separated by a multiplicative factor of at least $2^{\delta+10}$). We will say that cost classes are \textit{well separated}, and define $c_j= m/2^{j\cdot (\delta+10)}$.
    \item[(b)] All pairs in $\calP$ have contraction at most $\alpha$ when running $\alg$ on instance $\I$. We say that all pairs have \textit{low contraction}.
    \item[(c)] For any $i\geq 1$, the set of additional edges $E_i$ contains exactly one edge with the same endpoints of the pair $p_i$ and whose weight is exactly $\cost_{\alg}(\I,p_i)$ (i.e. we can assume $\alg$ connected the pair $p_i$ by simply using the single edge in $E_i$).
\end{enumerate}
\end{definition}
This definition suggests that we partition the set of pairs $\calP$ into \textit{cost classes} $\calP^{(1)},\ldots, \calP^{(j)},\ldots,\calP^{(\log(k)/(\delta+10))}$ where $\calP^{(j)}$ is the subset of pairs in $\calP$ that cost exactly $m/2^{(\delta+10) \cdot j}$. Note that there are at most $M\leq \log(k)/(\delta+10)$ distinct cost classes. Given this definition, we first claim the following lemma. Intuitively, the lemma states that worst-case instances can be reduced to $(\alpha,\delta)$-canonical instances (for some big $\delta$) at a multiplicative loss of $O(\log\log(k) + \log(\alpha))$.

\begin{lemma}
\label{lem:canonical_transform}
 For any instance $\I$ of size $k$, any greedy algorithm $\alg$ and any $\alpha\geq 1$, there exists an $(\alpha,\delta)$-canonical instance $\I'$ of size $k'\leq k$ such that 
 $$
     \cost_\alg(\I')\geq \frac{\cost_\alg(\I,\calP_{<\alpha})}{O (\log\log(k)+\log(\alpha))},
 $$
 $$
     \delta\geq 100\cdot (\log(\alpha)+\log\log(k)), \mbox{ and}
 $$
$$
    w(\OPT(\I'))\leq w(\OPT(\I)).
$$
\end{lemma}
\begin{proof} By standard geometric grouping arguments we can assume that there are at most $\log(k)$ cost classes $\calP^{(1)},\ldots ,\calP^{(\log(k))}$ such that the greedy algorithm $\alg$ pays a cost of $w(\OPT(\I))/2^{i-1}$ for all pairs in $\calP^{(i)}$. This first transformation already appeared in \cite{PanigrahiLecture} and loses a constant factor. Then we consider these cost classes but we keep only the pairs of contraction at most $\alpha$. Fix $\delta = 100\cdot (\log\log(k)+\log(\alpha))$ We partition the pairs as follows:
$$
    \tilde \calP^{(j)}=\bigcup_{i \equiv j \Mod{\delta+10}}\calP^{(i)}
$$ for all $0\leq j< (\delta+10)$. Since there are $(\delta+10)$ groups, one of them represents at least a fraction $1/(\delta+10)$ of the total cost. Keep only this group and transform the instance by adding additional edges to $\Se$ as follows. Assume that we kept the group $\tilde \calP^{(j)}$ then index the pairs in $\tilde \calP^{(j)}$ by order of arrival i.e. $\tilde \calP^{(j)}=\{p_1,\ldots , p_i,\ldots , p_{k'}\}$. For each pair $p_i\in \tilde \calP^{(j)}$, we define the set of additional edges $E^{(i)}$ as a single edge whose endpoints are exactly the endpoints of the pair $p_i$ and whose length is exactly what $\alg$ paid for this pair in the original instance $\I$. This formally describes the instance $\I'=(G,\tilde \calP^{(j)},\Se)$. We claim that for any pair selected, the greedy algorithm $\alg$ pays exactly the same cost for this pair regardless of which instance $\I$ or $\I'$ is running. We can prove this simple fact by induction of the number of pairs already arrived in $\I'$. If no pair has arrived this is clear. Now consider the next pair $p_i$ to arrive. Note that when running instance $\I'$, a path of length exactly $\cost_\alg(\I,p_i)$ is available to connect $p_i$ since  we added an edge in $E^{(i)}$ of exactly this length connecting the endpoints of $p_i$. We claim that there cannot be a shorter path. Indeed, by induction we assumed that $\alg$ paid the same in instance $\I$ and $\I'$ for previously arrived pairs hence it must be that the greedy algorithm $\alg$ used the additional edges in $\Se$ to connect previously arrived pairs. Because $\alg$ is \greedy for some $i\in \{1,2,3\}$, it must be that the shortcuts added by $\alg$ on instance $\I'$ so far are exactly edges of length 0 with endpoints at the endpoints of pairs arrived before $p_i$. Note that when running greedy on instance $\I$, the endpoints of previously arrived pairs must be at distance 0 when the new pair $p_i$ arrives. Hence all the shortcuts available to $\alg$ when receiving the pair $p_i$ in instance $\I'$ are also available when receiving the pair $p_i$ in instance $\I$. In particular, the shortest path taken by $\alg$ for pair $p_i$ in instance $\I'$ can only be longer than the path taken for pair $p_i$ in instance $\I$. 

One can see that in total we lose a multiplicative factor of at most $O(\delta)$ during the reduction. Finally, it is also clear that $w(\OPT(\I'))\leq w(\OPT(\I))$ and $k'\leq k$ since the graph $G$ has not changed and we keep in $\I'$ only a subset of the pairs in $\I$. This ends the proof of the lemma.
\end{proof}

The rest of the section will be devoted to the proof of the following theorem.
\begin{theorem}
\label{thm:main2}
Let $\alg$ be any greedy algorithm that uses one of our 3 contraction rules. Let $\I$ be an $(\alpha,\delta)$-canonical instance (of size $k$) of online Steiner Forest in decreasing metrics. Assume $\delta\geq 100\cdot (\log(\alpha)+\log\log(k))$. Then,
$$
    \cost_\alg \left(\I\right)\leq O(\log(k))\cdot  w(\OPT(\I)).
$$
\end{theorem}
Note that Theorem \ref{thm:main2} together with Lemma \ref{lem:canonical_transform} imply Theorem \ref{thm:main}. To see this, consider any instance $\I$. By losing a multiplicative factor of $O(\log\log(k)+\log(\alpha))$ and only considering the pairs in $\calP_{<\alpha}$, we transform the instance $\I$ into an $(\alpha,\delta)$-canonical instance $\I'$ using Lemma \ref{lem:canonical_transform}. Then we apply Theorem \ref{thm:main2} on instance $\I'$ and the total competitive ratio for pairs in $\calP_{<\alpha}$ will be $O(\log(k))\cdot O(\log\log(k)+\log(\alpha))$ which is exactly what we wanted to prove.

\paragraph{Dual fitting.} A key technical ingredient in the proof of Theorem \ref{thm:main2} will be dual fitting, which was also used in \cite{alon1992line,SODA96} and is a common technique in competitive analysis. In the case of Steiner Forest, a natural way to do dual fitting without explicitly writing a linear program is to consider a set of balls in the graph $G=(V,E)$. For some vertex $v\in V$ and some radius $r>0$, the ball $B(v,r)$ is the open ball of center $v$ and radius $r$, i.e.
$$
    B(v,r)=\{u\in V \mid d_{G}(v,u)<r\}.
$$
Denote by $\calT$ the set of \textit{terminals} which are the vertices that appear in at least one pair. For a terminal $u\in \calT$, denote by $\overline{u}$ its \textit{mate} which is the terminal that $u$ should be connected to. Note that we can assume without loss of generality that all terminals have a only one mate, by duplicating vertices if this is not the case. Assume we have a collection $\calB=\{B_1,B_2,\ldots ,B_\ell\}$ of balls such that:
\begin{itemize}
    \item All of the balls in $\calB$ are pairwise disjoint, and
    \item Each ball $B_i$ is centered at some terminal $u\in \calT$ and its radius $r_i$ satisfies $r<d_{G}(u,\overline{u})$.
\end{itemize}
Then if we define $y=\sum_{i}r_i$ the sum of radii of these balls it must be that
$$
w(\OPT(\I))\geq y.    
$$
A reason for this is that any feasible solution to the Steiner Forest instance must connect $u$ to $\overline{u}$. If we look at any ball $B_i\in \calB$, then at its center lies at a terminal $u$, and we know that $\overline{u}$ is not in $B_i$. Therefore, to connect $u$ to $\overline{u}$, a feasible solution needs to buy at least a path from the center to the border of $B_i$, which will have length at least $r_i$. Since all balls are pairwise disjoint, we know that these paths will be disjoint, and we can sum the lower bounds on each ball.

An alternative view of this is that the dual balls can be seen as a feasible solution (because the balls are pairwise disjoint) to the dual of the natural LP relaxation of the Steiner Forest problem. Then by weak duality, we know that any feasible solution has cost at least the cost of the dual. Hence we will also refer to a collection of balls as above as a \textit{dual solution}. A dual solution is \textit{feasible} if all the corresponding balls are pairwise disjoint and are all centered around the endpoints of some pairs. Using the proof technique of Awerbuch, Azar, and Bartal, we obtain the following lemma whose proof appears in Appendix \ref{sec:appendix_proofs}.
\begin{lemma}[\cite{SODA96}]
\label{lem:AAB}
Let $\I$ be an instance of online Steiner Forest in decreasing metrics. Consider a cost class $\tilde \calP\subseteq \calP$ of pairs. Let $\calP'\subseteq \tilde \calP$ be an arbitrary subset of $\tilde \calP$. Let $c$ be the constant such that 
$$
    \cost_\alg (\I,p)=c
$$ for all $p\in \tilde \calP$ and $\alg$ a greedy algorithm. Fix an arbitrary radius 
$$
    r\leq \frac{c}{8\log(|\tilde \calP|)}.
$$
Then for any such radius $r$ it is possible to construct a feasible dual solution $\calB$ such that:
\begin{enumerate}
    \item[(a)] All the balls $B\in \calB$ have a radius equal to $r$,
    \item[(b)] $|\calP'|\leq 5\cdot |\calB|$,
    \item[(c)] each pair $p\in \calP'$ has at most one ball $B\in \calB$ (denoted $B(p)$) centered around one of its endpoints, and
    \item[(d)] all balls in $\calB$ are centered around endpoints of pairs in $\calP'$.
\end{enumerate}
\end{lemma}
This lemma is actually the crux of the previous analysis that gives $O(\log^2(k))$ in general. We will use this lemma as a starting point for our improved analysis. We are now ready to start the overview of our main proof.

\subsection{Overview of the proof}
\label{subsec:mainoverview}

Recall that we aim to prove Theorem \ref{thm:main2}. By assumption we have that all pairs have contraction at most $\alpha$ and that the cost classes are well separated. Recall that this means the multiplicative gap between two consecutive cost classes is at least $2^\delta$ for some $\delta\geq 100\cdot (\log\log(k) + \log(\alpha))$. We use $M$ to denote the number of cost classes where class $j$ is denoted by $\calP^{(j)}$ for $1\leq j \leq M$.

The goal will be to construct a feasible dual solution $\calB$ that has some special properties. This dual solution will be constructed by taking a subset of the dual balls in each of the $\log(k)$ dual solutions constructed with the technique of Awerbuch, Azar, and Bartal \cite{SODA96}. In the end we will charge a portion of the cost that $\alg$ pays to the dual solution $\calB$. The remaining portion of the cost $\alg$ pays will be handled by an inductive argument. To this end we will have a charging scheme, $\ch:\calP\mapsto \mathbb{R}_+$ that will redistribute $cost_A(\I)$ amongst the terminal pairs.

Precisely,  the total cost that a pair $p$ carries will be
$$
    \ch(p)\cdot \cost_{\alg}(\I,p).
$$
Hence we see the charge as an additional multiplicative factor on the cost of a given pair (initially, the charge is set to 1 for all pairs). 

Note that we might sometimes decrease or increase the charge of a pair or transfer the cost, but we will always make sure that when a charge of a pair $p$ is decreasing, the charge of some other pairs are increasing accordingly so that no cost is lost. For any set $S$ of terminal pairs, we will let $\cost_\alg(\I,S,\ch)$ denote the total \textit{charged} cost carried by pairs in $S$. Formally,
$$
    \cost_\alg(\I,S,\ch) = \sum_{p\in S}\ch(p)\cdot \cost_\alg(\I,p).
$$

In the proof the pairs in $\calP$ will be classified into three types: \textit{surviving} pairs, \textit{charged} pairs, and \textit{dangerous} pairs. The \textit{surviving} pairs will contain pairs $p$ such that there is a ball $B(p)\in \calB$ centered around one of the endpoints. Intuitively these pairs are good for us since we can charge their cost directly to the dual ball $B(p)$. The other pairs will be by default classified as \textit{non-surviving}. Non-surviving pairs are further partitioned into two subsets, charged or dangerous pairs. \textit{Charged} pairs are those pairs that have their charge set to 0 (i.e. $\ch(p)=0$). Intuitively, they are also in an excellent situation for us since it means that we were able to transfer their cost entirely to other pairs. We do not need to count them in our total cost anymore. Finally, the \textit{dangerous} pairs are those pairs have neither a charge equal to $0$ nor a dual ball in $\calB$ centered at one of the endpoints. These pairs will be handled carefully via an inductive argument since we cannot charge them to the dual solution nor to some other pair. To keep careful track of all these elements, we will store a triple $(\calB,\ch,D)$ where $\ch$ is the charge function as described above, $\calB$ is a feasible dual solution and $D\subseteq \calP$ a set of \textit{dangerous} pairs. The family of dual balls will be a union of subsets of dual balls $\calB^{(j)}$ for $1\leq j \leq M$. Each of the balls in $\calB^{(j)}$ will account for a subset of pairs in $\calP^{(j)}$ and have some radius of roughly 
$$
    r_j=\frac{c_j}{8\log(k_j)},
$$ where $c_j$ is the cost of pairs in $\calP^{(j)}$ and $k_j=|\calP^{(j)}|$. This choice of radius is coming from previous work, summarized in Lemma \ref{lem:AAB}. 

Note that in our procedure, it might be that some pairs are not yet classified into one of the three categories (surviving, charged, or dangerous). However, at the beginning of iteration $j$, all pairs in cost classes $j'<j$ will be classified. The procedure contains two main steps.

\paragraph{Step 1.} In this step, we start taking into account interactions in-between cost classes. Informally, we do an iterative procedure from $j=1$ to $M$, where we try to build the dual solution from top to bottom. When we start iteration $j$ of this procedure we have a feasible dual solution composed of $\calB^{(1)}, \dots, \calB^{(j-1)}$ (of radius $\Theta(r_j)$ with $r_j$ specified above) centered around pairs in $\calP^{(1)}, \dots, \calP^{(j-1)}$. All the pairs in $\calP^{(j)}$ that are not yet classified are guaranteed to be far from the dual balls already in place. We then look at pairs in $\calP^{(j)}$ that are not yet classified, and build a dual solution $\calB^{(j)}$ around these pairs using Lemma \ref{lem:AAB}. Because unclassified pairs are far from previously placed balls in $\calB^{(1)}, \dots, \calB^{(j-1)}$, it is guaranteed that this new dual solution $\calB^{(j)}$ will not overlap with the previous dual solution. We then proceed as follows. For any ball $B\in \calB^{(j)}$ that we just added, we let $B_\calP$ denote the set of pairs of $\bigcup_{j'>j}\calP^{(j')}$ (note that we only consider pairs of smaller cost) such that one of its endpoints is at a distance at most
$$
    r\cdot \left(1+\frac{1}{200\cdot \log^2(K)} \right)
$$
from the center of $B$ (denoted $p(B)$). Here we note that for a technical reason, $K$ is only an upper bound on the real number of pairs $k$ (i.e. $K\geq k$). We will let $\partial B_\calP$ denote the set of pairs of $B_\calP$ that have one endpoint at a distance of at least 
$$
    r\cdot \left(1-\frac{1}{200\cdot \log^2(K)} \right)
$$
from the center of $B$. These pairs are on the border of $B$ hence the choice of notation. With a similar analogy, we will denote the interior of $B_\calP$ by $\mathring B_\calP$. This is the complement of $\partial B_\calP$ in $B_\calP$, i.e. 
$$
    \mathring B_\calP=B_\calP\setminus \partial B_\calP.
$$
Then for the current ball $B\in \calB^{(j)}$ at hand, centered at an endpoint of $p$, we look at the total charged cost of the pairs in $\mathring B_\calP$ and make a case distinction based on this value. If this cost is more than $\textrm{polylog}(K)$ times the charged cost of the pair $p$, then we can safely set the charge of $p$ to $0$, delete the dual ball $B$, and increase the charge of pairs in $B_\calP$ to account for this lost cost. Note that the charges of pairs in $\mathring B_\calP$ increase by at most a multiplicative $(1+\frac{1}{\textrm{polylog}(K)})$). This is pictured in the top left corner of Figure \ref{fig:aftercharge}. Since the number of cost classes is at most $M<\log(K)$ such accumulation of charges is not a problem (note that the charge of each pair can increase at most once per cost class in this way, since we only charge pairs inside a dual ball and the dual balls in a single cost class are pairwise disjoint). 

On the contrary, if the charged cost of pairs in $\mathring B_\calP$ is less than $\textrm{polylog}(K)$ times the charged cost of the pair $p$, we first halve the radius of $B$ to get $B'$. If the charged cost of the pairs inside $B'_\calP$ is at most a constant factor times the charged cost of $p$, we classify all the pairs in $B'_\calP$ as \emph{charged} and charge their cost to the pair $p$. Note that the charge of $p$ only increases by a constant factor when doing this, and this happens at most once per pair $p$ (when we place the dual ball around $p$). In addition we update $B$ to be $B'$, and we classify the pair $p$ as \emph{surviving}. If on the other hand, the charge of the pairs inside $B'_\calP$ is larger than a specified constant factor times the charge of $B$, we scale the radius of $B'$ up until we reach a point  where most of the cost in $B'_\calP$ is carried by $\mathring B'_\calP$ and not $\partial B_\calP'$. Then we mark all the pairs in $B'_\calP$ as \emph{dangerous} and add them to the set $D$. We update $B$ to be the ball $B'$ and classify $p$ as \emph{surviving}. This case is pictured in the bottom right corner of Figure \ref{fig:aftercharge}. Note that if the ball $B$ is not deleted, then all the pair in $B_\calP$ will be classified as either charged or dangerous. In particular, we will never try to place a dual ball around these pairs in the following iterations. We do this procedure for all the balls $B\in \calB^{(j)}$ and then move to iteration $(j+1)$. This step is handled in Subsection \ref{subsec:mainbalanced}.

\paragraph{Step 2.}  After Step 1, we end with a feasible dual solution $\calB$ consisting of the balls placed around surviving pairs, a charge function , and a set $D$ of dangerous pairs. Additionally, we guarantee that no pair is overcharged.

The pairs that are not dangerous are easily accounted for by the dual solution $\calB$. Indeed, the surviving pairs still have a dual ball around an endpoint, and the charged pairs have their cost entirely redistributed to other pairs. The only problem might come from dangerous pairs. However, because of how we constructed the dual solution $\calB$ and the set $D$, we will be able to cluster the dangerous pairs into disjoint sub-instances that are contained in dual balls corresponding to bigger cost classes. These instances are disjoint, and the crux of the argument is to show a statement of the form:

\vspace{0.1cm}
\textit{If the greedy algorithm were to run separately on each sub-instance, then the cost greedy would pay for these pairs would be the same cost that it was paying for these pairs in the bigger instance $\I$.}

Hence we can argue that the total cost incurred for dangerous pairs is at most the sum of costs paid by greedy on each sub-instance separately. This helps because we only put pairs in $D$ in the case that their charged cost was bounded by $\textrm{polylog}(K)$ times the charged cost of the pair that created the ball $B$ that contains them. As a result, we have a strong upper bound on the number of pairs $k'$ in each smaller sub-instance $\I'$. To finish the proof, we need to bound the cost of the offline optimum for each sub-instance $\I'$. We note that because all the pairs in $D$ are in the \textit{interior} of $B$ (i.e. far from the border of $B$), we can modify the metric of the graph $G$ at the border of $B$. This will not change the behavior of greedy for the pairs in $D$ because the border is way too far from the interior of $B$ for greedy to be tempted to use the modified metric (recall that greedy always takes the shortest path). We will define a new graph $G'$, which is the graph induced by vertices in $B$. We also say that all the vertices \textit{exactly} on the border of $B$ are all at a distance $0$ from each other (see bottom right corner of Figure \ref{fig:aftercharge}). With this modification, it becomes clear that the offline optimum cost on instance $\I'$ is at most the cost paid by $\OPT(\I)$ inside $B$, which we will denote by $w(\OPT(\I)\cap B)$. Hence the offline optimum cost for each sub-instance is at most what the global optimum pays locally inside the ball that created the sub-instance. Since all balls in $\calB$ are disjoint, these areas never overlap; hence the \textit{sum} of all local optima is at most the global optimum of instance $\I$. Using this observation, we handle the cost incurred by pairs in $D$ via a delicate induction hypothesis on the number of cost classes in the instance. This induction is described formally in Subsection \ref{subsec:maininductive}.

\subsection{Building a balanced dual solution}
\label{subsec:mainbalanced}
We formalize here Step 1 of the previous subsection. We give a formal definition of all the properties that our triple $(\calB,\ch,D)$ should satisfy. Note that we are also given an upper bound $K$ on the real number of pairs $k$ (this is for technical reasons for handling the induction in the next subsection). Recall that for a ball $B\in \calB^{(j)}$ we denote by $B_\calP$ the set of pairs of $\bigcup_{j'>j}\calP^{(j')}$ such that one of its endpoint is at distance at most
$$
    r\cdot \left(1+\frac{1}{200\cdot \log^2(K)} \right)
$$ from the center of $B$. We also have similar definitions for $\partial B_\calP$ and $\mathring B_\calP$. Now we can state the main definition of this subsection. Intuitively, conditions (a) and (b) state that $\calB$ is a feasible dual solution whose dual balls have radii large enough. Condition (c) states that the total charged cost of dangerous pairs inside the ball $B$ is never much more than $\mathrm{polylog}(K)$ times the charged cost of the pair that created the ball $B$. Similarly, condition (d) states that the charged cost of dangerous pairs on the border of $B$ is not more than $1/\log(K)$ times the charged cost of dangerous pairs strictly inside $B$. The last condition (e) states that no pair was charged too many times.  
\begin{definition}[Balanced dual solution]
\label{def:balanced}
 A balanced dual solution for an $(\alpha,\delta)$-canonical instance $\I$ with respect to algorithm $\alg$ is a quadruple $(\calB,\ch,D,K)$ such that:
 \begin{enumerate}
     \item[(a)] All balls $B\in \calB$ are pairwise disjoint and $D\subseteq \bigcup_{B\in \calB} B_\calP$. Moreover, $\calB$
 is partitioned into $M$ sub-collections of balls $\calB^{(1)},\calB^{(2)},\ldots ,\calB^{(M)}$ such that,
    \item[(b)] For every $j\geq 1$, each ball in $\calB^{(j)}$ has a radius $r'$ that satisfies $r_j/2\leq r' \leq r_j=c_j/8\log(k_j)$, with $c_j$ the cost associated to cost class $\calP^{(j)}$ and $k_j=|\calP^{(j)}|$.
    \item[(c)] For every ball $B\in \calB^{(j)}$
    $$
        \cost_\alg(\I,\mathring B_\calP\cap D,\ch)\leq 10\cdot \ch(p(B))\cdot c_j \cdot \log^{10}(K),
    $$
    \item[(d)] For any ball $B\in \calB^{(j)}$
    $$
        \cost_\alg(\I,\partial B_\calP\cap D,\ch)\leq \frac{10\cdot \cost_\alg(\I,\mathring B_\calP\cap D,\ch)}{\log(K)},
    $$
    \item[(e)] For any $j>0$, any pair $p\in \calP^{(j)}$,
    $$
        \ch(p) \leq 
    \begin{cases}
     55\cdot e^5 \mbox{ if $p$ is a surviving pair},\\
     0\mbox{ if $p$ is a charged pair},
    \end{cases}
    $$
    $$
        \ch(p) \leq \left(1+\frac{5}{\log(K)} \right)^{j-1}\mbox{ if $p$ is a dangerous pair that belongs to $B_\calP$ for some $B\in \calB^{(j)}$}, and
    $$
    $$
        \ch(p) \geq 1\mbox{ if $p$ is a surviving or dangerous pair.}
    $$
 \end{enumerate}
\end{definition}

The main result of this subsection will be that it is always possible to find a balanced dual solution.
\begin{lemma}
\label{lem:balanced}
Given an $(\alpha,\delta)$-canonical instance $\I$ with respect to $\alg$, a balanced dual solution $(\calB,\ch,D,K)$ always exists provided that $\delta\geq 100\cdot (\log(\alpha)+\log\log(K))$, the number of cost classes $M$ satisfies $M\leq \log(K)$, and the number of pairs $k$ satisfies $k\leq K$.
\end{lemma}
\begin{proof} We build the solution quadruple $(\calB,\ch, D,K)$ with an iterative procedure from $j=1$ to $j=M$ (recall that $M$ is the number of cost classes in $\I$) that will maintain the following invariants at the beginning of any iteration $j$:
\begin{enumerate}
    \item[(i)] For any $j'<j$, the dual balls in $\calB^{(j')}$ are already fixed and all pairs in $\calP^{(j')}$ are already classified as either surviving, charged, or dangerous. The dual balls of $\calB^{(j')}$ satisfy conditions (b), (c), (d) of Definition \ref{def:balanced}.
    \item[(ii)] For any $j'<j$, the charge of pairs in $\calP^{(j')}$ satisfy condition (e) of Definition \ref{def:balanced}.
    \item[(iii)] For any $j'\geq j$, the pairs in $\calP^{(j')}$ can be classified as either dangerous or charged, or not be classified yet. No pair of $\calP^{(j')}$ is classified as surviving yet. We also have $\calB^{(j')}=\emptyset$.
    \item[(iv)] For any $j'\geq j$, all pairs in $\calP^{(j')}$ classified as either dangerous or charged satisfy condition (e) in Definition \ref{def:balanced}. The pairs $p$ that are not yet classified satisfy
    $$
        1\leq \ch(p)\leq \left(1+\frac{5}{\log(K)} \right)^{j-1}.
    $$
    \item[(v)] For any $j'\geq j$, all pairs in $\calP^{(j')}$ not yet classified do not belong to any set $B_\calP$ for some already existing ball (i.e. unclassified pairs are far from the balls already placed).
\end{enumerate}
We start iteration $j$ by considering the pairs in $\calP^{(j)}$ that are not yet classified. To avoid confusion, let us denote by $\calP'^{(j)}\subseteq \calP^{(j)}$ this set of unclassified pairs. Using Lemma \ref{lem:AAB}, we get a dual solution $\calB^{(j)}$ of balls all of radius
$$
    r_j=\frac{c_j}{8\log(k_j)}
$$
that are all centered around endpoints of pairs in $\calP'^{(j)}$ and such that no pair in $\calP'^{(j)}$ has more than one ball centered around an endpoint.  All the pairs in $\calP'^{(j)}$ that do not have a dual ball can already be classified as charged. We decrease the charge of these pairs to 0 increase the charge of the other pairs in $\calP'^{(j)}$ to compensate. Since we have that $|\calP'^{(j)}|\leq 5\cdot |\calB^{(j)}|$ and invariant (iv), we have that the charge of the remaining pairs in $\calP'^{(j)}$ will be at most

\begin{equation}
\label{eq:inc_charge_delete_large_ball}
    5\cdot \left(1+\frac{5}{\log(K)} \right)^{(j-1)}\leq 5\cdot \left(1+\frac{5}{\log(K)} \right)^{M}\leq 5\cdot e^{5}
\end{equation}
after this step (we use the assumption that $M\leq \log(K)$ and $(1+x)\leq e^x$). Then for each new ball $B\in  \calB^{(j)}$ that we created around the endpoint of a pair $p$ we consider the following process with two cases. We consider the pairs of $\mathring B_\calP$ (i.e. pairs of smaller cost in the \textit{interior} of $B$) and proceed to a case distinction on the value of their total charged cost 
$$
    \cost_\alg(\I,\mathring B_\calP,\ch)
$$ that we denote $\Sigma$.

\paragraph{Case 1 (Top left corner in Figure \ref{fig:aftercharge}).} If 
\begin{equation}
\label{eq:lemma_balanced_charge1}
    \Sigma > 10\cdot \ch(p)\cdot c_j \cdot \log^{10}(K),
\end{equation}
that we can rephrase intuitively  as "$B$ does not satisfy condition (c)" then we simply erase the dual ball $B$ and charge all the cost $\ch(p)\cdot c_j$ to the pairs in $\mathring B_\calP$ proportionally to their weight. Note that by doing this, all the pairs in $\mathring B_\calP$ see their charge increasing by a multiplicative $\left(1+\frac{1}{(10\log^{10}(K))}\right)<\left(1+\frac{5}{\log(K)} \right)$. The pair $p$ is now classified as charged. Note that the pairs in $\calP^{(j'>j)}$ have now a charge that is at most $\left(1+\frac{5}{\log(K)} \right)^{j}$ by applying invariant (iv). This will be the only place where their charge can increase during iteration $j$ and note that this happens only once per iteration since only the pairs inside a ball $B\in \calB^{(j)}$ are charged and the balls in $\calB^{(j)}$ are pairwise disjoint. Therefore invariant (iv) will hold at the beginning of iteration $j+1$ for these pairs.

\paragraph{Case 2.} In this case we assume that $\Sigma \leq  10\cdot \ch(p(B))\cdot c_j \cdot \log^{10}(K)$. We consider the ball $B'$ with the same center as $B$ but with only half the radius of $B$ (i.e. $B'$ is a ball centered at $p(B)$ with radius $r=r_j/2$). We look at the total charged cost of pairs in $B'_\calP$ and proceed to a case distinction on this value (note that we do not consider only the strict interior of $B$ but also its border). Denote by $\Sigma'$ this new value, i.e.
$$
    \Sigma' = \cost_\alg(\I,B'_\calP,\ch).
$$
We proceed by a sub-case distinction of the value of $\Sigma'$.

\paragraph{Sub-case 2(a).} If $\Sigma'\leq 10\cdot \ch(p(B))\cdot c_j$ then we simply charge the cost of all the pairs in in $B'_\calP$ to the pair $p$. In this case the pairs in $B'_\calP$ become charged pairs. Note that the charge of $p$ in that case is multiplied by at most $11$. We also set the dual ball $B$ to $B'$ (i.e. we scale down the radius of $B$ by a factor 2). By Equation \ref{eq:inc_charge_delete_large_ball}, we get that the charge of $p$ will be at most
$$
    \ch(p)\leq 11 \cdot 5\cdot e^5 = 55 e^5.
$$ This will be the final charge of this pair hence invariant (ii) will be satisfied.

\paragraph{Sub-case 2(b) (Bottom right corner in Figure \ref{fig:aftercharge}).} If $\Sigma'> 10\cdot \ch(p(B))\cdot c_j$, then we will try to increase the radius $r$ of $B'$ by an increment of
$$
    \Delta r = \frac{r_j}{200\cdot \log^2(K)},
$$
until we find a radius which satisfies condition (d) for the ball $B'$. We denote by $B^{(t)}$ the ball obtained after increasing the radius of $B'$ $t$ times by $\Delta r$. Hence the radius of $B^{(t)}$ is $r_t=r_j/2+t\cdot (\Delta r)$. Each time we increase $r$, we check if 
 \begin{equation}
 \label{eq:radiuschoice}
        \cost_\alg(\I,\partial B^{(t)}_\calP,\ch)\leq \frac{10\cdot \cost_\alg(\I,\mathring B^{(t)}_\calP,\ch)}{\log(K)}.
\end{equation}
If this is the case we stop and add all the pairs in $B^{(t)}_\calP$ to $D$. The pairs in $B^{(t)}_\calP$ become in this case dangerous pairs. If not we continue increasing the radius $r$ until Equation \ref{eq:radiuschoice} holds. We claim that this process will stop before reaching $r_t=r_j$. To see this we first prove that $ B^{(t)}_\calP \cap \partial B^{(t+6)}_\calP=\emptyset$ for any $t\geq 0$. Indeed if a pair $p_i\in \calP^{(j')}$ with $j'>j$ belongs to $ B^{(t)}_\calP \cap \partial B^{(t+6)}_\calP$ then it must be that one of the two endpoints of $p_i=\{s_i,t_i\}$ belongs to $B^{(t)}$ while the other belongs to the border $\partial B^{(t+6)}$ since $B^{(t)}\cap \partial B^{(t+6)} = \emptyset$. Hence the distance between the two endpoints of $p_i=\{s_i,t_i\}$ satisfies:
$$
    d_G(s_i,t_i)\geq \left(r_t+6\Delta r \right)\cdot \left(1-\frac{1}{200\cdot \log^2(K)} \right)-r_t\cdot  \left(1+\frac{1}{200\cdot \log^2(K)} \right)\geq 5\Delta r-\frac{2r_t}{200\cdot \log^2(K)}\geq \frac{r_j}{200\cdot \log^2(K)}. 
$$ where the last inequality comes from the fact that $r_t\leq r_j$. However, we assumed that the instance $\I$ is $(\alpha,\delta)$-canonical with $\delta\geq 100\cdot (\log(\alpha)+\log\log(K))$ and this implies that all pairs have contraction at most $\alpha$ which means that the pair $p_i$ should have costed $\alg$ at least 
$$
    \frac{r_j}{\alpha (200\cdot \log^2(K))} \geq \frac{c_j}{\alpha (8\cdot 200\cdot \log^3(K))}>\frac{c_j}{2^{10}\cdot \alpha^{100}\log^{100}(K)}>\frac{c_j}{2^{\delta+10}}=c_{j+1},
$$
which is a contradiction. Thus $B^{(t)}_\calP \cap \partial B^{(t+6)}_\calP=\emptyset$. Since we assume the process does not stop, we have that 
 $$
     \cost_\alg(\I,\partial B^{(t+6)}_\calP,\ch)> \frac{10\cdot \cost_\alg(\I,\mathring B^{(t+6)}_\calP,\ch)}{\log(K)}\geq \frac{10\cdot \cost_\alg(\I, B^{(t)}_\calP,\ch)}{\log(K)}.
 $$ Together with the fact that $B^{(t)}_\calP \cap \partial B^{(t+6)}_\calP=\emptyset$ we get 
$$
     \cost_\alg(\I,B^{(t+6)}_\calP,\ch)\geq \cost_\alg(\I,B^{(t)}_\calP,\ch)+ \cost_\alg(\I,\partial B^{(t+6)}_\calP,\ch)\geq \cost_\alg(\I,B^{(t)}_\calP,\ch)\cdot \left(1+\frac{10}{\log(K)} \right).
 $$
 This implies that, for all $t\geq 0$,
 $$
 \cost_\alg(\I, B^{(t)}_\calP,\ch) \geq \cost_\alg(\I, B^{(0)}_\calP,\ch)\cdot \left(1+\frac{10}{\log(K)}\right)^{\left\lfloor t/6 \right\rfloor} = \Sigma'\cdot \left(1+\frac{10}{\log(K)}\right)^{\left\lfloor t/6 \right\rfloor}.
 $$
However, we have by assumption in our case that
$$
    \frac{\Sigma}{\Sigma'} \leq \log^{10}(K),
$$
thus for $t=60\log(K) \log\log(K)$, we can write simultaneously
\begin{equation}
\label{eq:balancedlemma2}
    \cost_\alg(\I, B^{(t)}_\calP,\ch) \geq \Sigma'\cdot \left(1+\frac{10}{\log(K)}\right)^{10 \log(K)\log\log(K)}> \Sigma' \cdot \log^{10}(K) \geq \Sigma,
\end{equation}
and
\begin{equation}\label{eq:balancedlemma3}
    r_t = r_j/2+t\cdot (\Delta r) = r_j\cdot \left(\frac{1}{2}+ \frac{60\log(K) \log\log(K)}{200\cdot \log^2(K)}\right)<r_j.
\end{equation}
This is a contradiction since by Equation \ref{eq:balancedlemma3} we should have $B^{(t)}_\calP\subseteq B_\calP$ and by Equation \ref{eq:balancedlemma2} we have $\cost_\alg(\I, B^{(t)}_\calP,\ch)>\cost_\alg(\I, B_\calP,\ch)$. Thus the process must stop before reaching $t=60\log(K) \log\log(K)$, hence before reaching $r_t=r_j$.

 \paragraph{Correctness.} We note that Sub-case 2(b) is the only case in which we create dangerous pairs. By construction, properties (c) and (d) will hold for any ball. This is because if a ball $B$ remains around a pair $p$ at the end of the procedure, either all the pairs inside are charged to $p$, or they become dangerous pairs. In the case where the pairs inside become dangerous the procedure described in Sub-case 2(b) stops exactly when properties (c) and (d) are satisfied. Property (b) is also satisfied since the only place where a radius can be modified is in Sub-case 2(b). In this case, we show that after halving the radius, it cannot grow for too long before the number of pairs in the border is less concentrated than in the interior. This means that since the balls in $\calB^{(j)}$ were disjoint when taken with radius $r_j$, they also have to be disjoint if the radius is slightly smaller. Condition (a) is also satisfied. To see this, note that the balls in $\calB^{(j)}$ cannot intersect balls installed before since in Case 2, which is the only case where a ball survives, for any ball $B\in \calB^{(j'<j)}$, all the pairs in $B_\calP$ become classified as either dangerous or charged (invariant (v)). In particular, we will never try to place a ball around these pairs since we only place dual balls around pairs that are not classified yet. Since $B_\calP$ encompasses a region slightly bigger than the ball $B$, the smaller dual balls will not intersect $B$. We also have that $D\subseteq \bigcup_{B\in \calB} B_\calP$ since dangerous pairs are only created in Sub-case 2(b) where we consider pairs inside a ball $B$. It remains to show that condition (e) holds. Note that for charged pairs, this is clear. For surviving pairs, we showed that our process maintains invariant (ii); hence after the end of the procedure, condition (e) is satisfied. For dangerous pairs, note that once a pair becomes dangerous, its charge will never increase anymore. Hence this shows that invariant (ii) also holds for these pairs, and in particular, at the end of the process, condition (e) holds.
 \end{proof}

\subsection{Inductive proof using balanced dual solutions}
\label{subsec:maininductive}

Given the previous subsection, we are ready to state the main induction. Recall that in a balanced dual solution $(\calB,\ch,D)$, all the pairs except the dangerous pairs (the set $D$) are accounted for by the dual balls in $\calB$. Our induction will precisely take advantage of this. In the following, for any ball $B$ and instance $\I$, we denote by $w\left(\OPT(\I)\cap B\right)$ the total cost of edges that are contained in $B$ and bought by $\OPT(\I)$ (note that we assume that edges in $G$ are arbitrarily small, so no edge is crossed by the border of $B$). Note that we also have the straightforward bound $r\leq w\left(\OPT(\I)\cap B\right)$ (with $r$ being the radius of $B$) since the optimum solution needs to connect at least the center of the ball $B$ to its border. For a collection of disjoint balls $\calB$, we define
$$
    w\left(\OPT(\I)\cap \calB\right)=\sum_{B\in \calB} w\left(\OPT(\I)\cap B\right).
$$

\begin{lemma}
\label{lem:main_induction}
Let $\I$ be an $(\alpha,\delta)$-canonical instance with respect to algorithm $\alg$. Assume that its size is $k\leq K$, and $\delta\geq 100\cdot (\log(\alpha)+\log\log(K))$. Assume there are $M\leq \log(K)$ distinct cost classes when running $\alg$ on instance $\I$. Let $\left(\calB=\bigcup_{j=1}^M\calB^{(j)},\ch,D,K\right)$ be a balanced dual solution. Then we have that
$$
    \cost_\alg(\I) \leq (880 e^{5})\cdot  \left( \sum_{j=1}^M \log(k_j)\cdot w\left(\OPT(\I)\cap \calB^{(j)}\right)\right) + e^{200+20M/\log(K)}\cdot \left(\sum_{j=1}^{M} \delta \cdot (M-j)\cdot w\left(\OPT(\I)\cap \calB^{(j)}\right)\right),
$$
where $k_j$ is the number of terminals in cost class $j$ for $j=1, \dots, M$.
\end{lemma}
Intuitively, the first term of the right-hand side of the inequality corresponds to the pairs that are either surviving or charged and can be charged to the dual. The second term corresponds to the cost paid by $\alg$ because of dangerous pairs in $D$. The proof of Lemma \ref{lem:main_induction} will be done by induction on the number of cost classes $M$.
\begin{proof}
The base case of the induction is for $M=0$ in which case the statement is vacuously true (the instance is empty). Hence assume $M>0$ and that the statement is true for any $(\alpha,\delta)$-canonical instance with $M'<M$ cost classes, and satisfying the conditions of the lemma.

Recall that in a balanced dual solution $(\calB=\bigcup_{j=1}^M \calB^{(j)},\ch,D,K)$ (that exists by Lemma \ref{lem:balanced}) we have that $\ch (p)\leq 55e^{5}$ for all surviving pairs $p$, that $\calB$ is a feasible dual solution, and that the radius of each ball $B\in \calB^{(j)}$ is at least $1/(16\log(k_j))$ times the cost of corresponding pair. Hence we have that the total cost of surviving or charged pairs (denote this set $\calP'$) is at most
$$
    \cost_\alg(\I,\calP') \leq (16\cdot 55 e^{5})\cdot  \left( \sum_{j=1}^M \log(k_j)\cdot w\left(\OPT(\I)\cap \calB^{(j)}\right)\right) = (880 e^{5})\cdot  \left( \sum_{j=1}^M \log(k_j)\cdot w\left(\OPT(\I)\cap \calB^{(j)}\right)\right).
$$

All that remains to do is to upper bound the total cost incurred by $\alg$ because of pairs in $D$. By property (a) of Definition \ref{def:balanced}, it suffices to consider each ball $B\in \calB$ and the dangerous pairs it contains. Hence fix a ball $B\in \calB^{(j)}$ of radius $r$. We consider the instance $\I'=(G',\calP',\Se')$ that is defined as follows (see bottom right corner of Figure \ref{fig:aftercharge}).
\begin{itemize}
    \item The graph $G'$ is the graph induced by $G$ on the vertex set $B$, in which we add an edge of cost $0$ between any pair of vertices at a distance exactly $r$ from the center of $B$. We denote by $E'$ this set of additional edges of weight 0 (see Figure \ref{fig:aftercharge} in the bottom right corner).
    \item $\calP'$ is the set of pairs in $\mathring B_\calP\cap D$, given in the same relative order to algorithm $\alg$. Note that these pairs must have both endpoints inside $B$ because all pairs have a low contraction, and the cost paid by these pairs is much smaller than the radius of the ball. Hence traveling from the interior of $B$ to the border is already way too far. We also emphasize that we ignore pairs in $\partial B_\calP\cap D$.
    \item The set of additional edges $\Se'$ contains all the edges in $\Se$ that we revealed just before reading a pair in $\calP'$. Moreover, the edges in $\Se'$ are revealed in the same way to $\alg$, that is, if an edge $e\in \Se$ is revealed just before a pair $p \in \calP'$, then it is also revealed just before $p$ in instance $\I'$.
\end{itemize}

We then make the following claim.
\begin{claim}
\label{cla:subinstance}
Given instance $\I'$ we have that:
\begin{enumerate}
    \item[(a)] For all $p\in \calP'$, $\cost_\alg(\I',p)=\cost_\alg(\I,p)$. In particular, the cost incurred by $\alg$ because of pairs in $\calP'$ is \textit{exactly} the same as the cost that $\alg$ would pay when running on instance $\I'$.
    \item[(b)] $w\left(\OPT(\I') \right)\leq w\left(B\cap \OPT(\I) \right)$.
\end{enumerate}
\end{claim}
\begin{proof}
To see (b), let us buy the edge set $(\OPT(\I)\cap B)\cup E'$. It is clear that this costs at most $w\left(B\cap \OPT(\I) \right)$ since all edges in $E'$ have cost 0. We claim that this is a feasible solution for the instance $\I'$. Consider any pair $p=\{s,t\}\in \calP'$. It must be that $\OPT(\I)$ contains a path between $s$ and $t$. If this path does not leave the ball $B$ then it is contained in the edge set $(\OPT(\I)\cap B)\cup E'$. Otherwise, denote this path by its sequence of vertices:
$s,v_1,v_2,\ldots ,v_k, t$. Denote by $v_f$ and $v_\ell$ the first and last vertices on this path that are outside $B$. Then it is clear that the edge set $(\OPT(\I)\cap B)\cup E'$ contains the path $s,v_1,\ldots, v_{f-1},v_{\ell+1},\ldots ,v_k, t$. In both case, $s$ and $t$ are connected and $(\OPT(\I)\cap B)\cup E'$ is a feasible solution to instance $\I'$.

To see (a), we prove by induction on the number of pairs $\calP'$ arrived so far that the claim holds. If no pair arrived yet, this is vacuously true. Otherwise, consider the arrival of pair $p$ and assume (a) holds for all pairs that arrived before. First note that the set of shortcuts added by $\alg$ when connecting any pair $p'=\{s,t\}$ in instance $\I$ is just one edge $e=\{s,t\}$ with weight 0. This is because we assumed that just before $p'$ arrived, an edge of $\Se$ between $s$ and $t$ arrived with weight exactly $\cost_\alg(\I,p')$. Hence we can assume that $\alg$ went through this edge to connect $p'$. Here we also use the fact that when $\alg$ goes through only one edge to connect a pair, the three contraction rules behave exactly the same. 

Now by induction hypothesis, the set of shortcuts bought by $\alg$ so far on instance $\I'$ have exactly the same property. Hence this set of shortcuts is a subset of the shortcuts added by $\alg$ on instance $\I$. We claim that when pair $p=\{s,t\}$ arrives, the shortest path is again through the corresponding edge of $\Se'$. Assume this is not the case. As a shorthand, define $c=\cost_\alg(\I,p)$. Denote by $\mathring B$ the set of vertices in $B$ that are at distance at most 
$$
    r\cdot \left(1-\frac{1}{100\cdot \log^2(K)} \right)
$$ from the center of $B$. We first claim that all the shortcuts added so far by $\alg$ in instance $\I'$ have both endpoints in $\mathring B$ (i.e. far from the border of $B$). To see this, note that all the pairs in $\calP'$ are in $\mathring B_\calP$  hence they have one endpoint at distance at most 
$$
    r\cdot \left(1-\frac{1}{200\cdot \log^2(K)} \right)
$$ from the center of $B$. If the other endpoint was not in $\mathring B$, the distance $d$ between the two endpoints would be at least
$$
    d\geq \frac{r}{200\cdot \log^2(K)},
$$ and since pairs have contraction at most $\alpha$, we would have that the cost paid for this pair in $\I$ is at least 
$$
    \frac{d}{\alpha}\geq \frac{r}{200\cdot (\alpha\log^2(K))}
$$
but we assume that cost classes would be separated by at least a multiplicative $2^{\delta+10}>2^{10}\cdot \alpha^{100}\cdot \log^{100} (K)$ by assumption on $\delta$. Hence we would have a contradiction. Hence all pairs in $\calP'$ have both endpoints in $\mathring B$, but this also implies that all shortcuts added so far by $\alg$ in instance $\I'$ have both endpoints in $\mathring B$.

Now remark that because the instance is $(\alpha,\delta)$-canonical, we have that $d_G(s,t)\leq \alpha c$, and since the distance between a point in $\mathring B$ and the exact border of $B$ is much bigger than $\alpha c$, it cannot be that the shortest path uses some edges in $E'$ (just to reach them is already too expensive because all previously bought shortcuts have endpoints in $\mathring B$). Hence the shortest path only uses edges that were available to $\alg$ when connecting $p$ in instance $\I$. Therefore the shortest path available in $\I'$ can only be longer than the shortest path for the same pair in $\I$. Since the corresponding edge of $\Se'$ of weight exactly $c$ is available, the shortest path is in fact exactly the same.
\end{proof}

By Claim \ref{cla:subinstance}, we know that the cost incurred by $\alg$ because of pairs in $\mathring B_\calP$ is equal to the cost that $\alg$ would pay on the instance $\I'$. Denote by $k'_i$ the number of pairs in $\calP'^{(i)}$ in $\I'$. Note that by Claim \ref{cla:subinstance} we have $\calP'^{(i)}=\calP^{(i)}\cap \mathring B_\calP\cap D$. Because of properties (c) and (e) of Definition \ref{def:balanced}, it must be that the following inequalities hold
\begin{equation}
\label{eq:induction1}
    \cost_\alg(\I,\mathring B_\calP,\ch) \leq 10\cdot \ch(p(B)) \cdot c_j \cdot \log^{10}(K) \leq (550e^5) \cdot c_j \cdot \log^{10}(K)
\end{equation} where $c_j$ is the cost of the pair that created the ball $B$, and
\begin{equation}
\label{eq:induction2}
    k'_i\cdot c_i\leq \cost_\alg(\I,\mathring B_\calP,\ch).
\end{equation}
Putting Equations \ref{eq:induction1} and \ref{eq:induction2} together we obtain 
\begin{equation}
    \label{eq:induction3}
    k'_i\leq (550e^5) \cdot \log^{10}(K) \cdot \frac{c_j}{c_i} = (550e^5) \cdot \log^{10}(K) \cdot 2^{(\delta+10)\cdot (i-j)}\leq 2^{3(\delta+10)\cdot (i-j)},
\end{equation} for all $i>j$.

Note that in instance $\I'$, we have at most $M'=M-j<M$ cost classes. Finally, it is clear that $\I'$ is still an $(\alpha,\delta)$-canonical instance , $\alg$ behaves the same on the relevant pairs, and all conditions of Lemma \ref{lem:main_induction} are satisfied (we keep the same upper bound $K$ on the number of pairs). Hence we can apply the induction hypothesis on the sub-instance $\I'$ to obtain the following upper bound on the cost of dangerous pairs inside $B$. 
$$
    \cost_\alg(\I')\leq (880e^5) \left( \sum_{i=1}^{M'} \log(k'_i)\cdot w\left(\OPT(\I')\cap \calB'^{(i)}\right)\right) + e^{200+20M'/\log(K)}\cdot \left(\sum_{i=1}^{M'} \delta \cdot (M'-i)\cdot w\left(\OPT(\I')\cap \calB'^{(i)}\right)\right),
$$ with $k'_i\leq  2^{3(\delta+10)\cdot i}$ by re-indexing cost classes from 1 to $M'$ in Equation \ref{eq:induction3}. The first term on the right-hand side is

\begin{align*}
    &(880e^5)\cdot  \left( \sum_{i=1}^{M-j} \log(k'_i)\cdot w\left(\OPT(\I')\cap \calB'^{(i)}\right)\right)\\
    &\leq (880e^5)\cdot \left( \sum_{i=1}^{M-j} \log\left(2^{3(\delta+10)\cdot i}\right)\cdot w\left(\OPT(\I')\cap \calB'^{(i)}\right)\right)\\
    &\leq (2640e^5)\cdot (\delta+10)\cdot \sum_{i=1}^{M-j} i\cdot w\left(\OPT(\I')\cap \calB'^{(i)} \right)\\
    &\leq \left(\delta e^{200+20M'/\log(K)}\right)\cdot  \sum_{i=1}^{M-j} i\cdot  w\left(\OPT(\I')\cap \calB'^{(i)} \right).
\end{align*}

The second term in the right-hand side is less than
\begin{align*}
    &\left(e^{200+20M'/\log(K)}\right) \cdot \left(\sum_{i=1}^{M-j} \delta \cdot(M-j-i)\cdot w\left(\OPT(\I')\cap \calB'^{(i)}\right)\right).
\end{align*}
By summing both terms, we obtain a cost of at most 
\begin{align*}
    &\left(e^{200+20M'/\log(K)}\right)\cdot  \sum_{i=1}^{M-j} (i\delta) \cdot w\left(\OPT(\I')\cap \calB'^{(i)} \right) + \\
    &\left(e^{200+20M'/\log(K)}\right) \cdot \left(\sum_{i=1}^{M-j} \delta \cdot(M-j-i)\cdot w\left(\OPT(\I')\cap \calB'^{(i)}\right)\right)\\
    &\leq \left(\delta e^{200+20M'/\log(K)}\right)\cdot \left(\sum_{i=1}^{M-j} (M-j)\cdot w\left(\OPT(\I')\cap \calB'^{(i)}\right)\right).
\end{align*}

Now recall that in a balanced dual solution, all balls are pairwise disjoint hence we obtain the upper bound
\begin{equation}
\label{eq:induction4}
    \left(\delta e^{200+20M'/\log(K)}\right)\cdot  (M-j)\cdot  w(\OPT(\I'))
\end{equation}
on the cost of pairs in $\mathring B_\calP\cap D$ when summing on all $i$ in the previous inequality.

By Claim \ref{cla:subinstance}, we also have $w(\OPT(\I'))\leq w\left(B\cap \OPT(\I) \right)$. Now we need to take into account the charges that were put on dangerous pairs in $\mathring B_\calP\cap D$ by the balanced dual solution in the big instance $\I$. By property (d) of Definition \ref{def:balanced}, the total charged cost incurred by $\alg$ for pairs in $\partial B_\calP \cap D$ is a most $10/\log(K)$ fraction of the cost of pairs in $\mathring B_\calP\cap D$. Additionally, the charge of each pair in $p\in \mathring B_\calP\cap D$ is at most (by property (e))
$$
    \left(1+\frac{5}{\log(K)} \right)^{j-1}\leq e^{5(j-1)/\log(K)},
$$ if the ball that contains this pair is in $\calB^{(j)}$. Hence we only need to increase the upper bound given by Equation \ref{eq:induction4} by a multiplicative term $e^{(5(j-1)+10)/\log(K)}\leq e^{20j/\log(K)}$ to get a final upper bound of 
\begin{equation}
\label{eq:induction5}
    \left(\delta e^{200+(20(M-j)+20j)/\log(K)}\right)\cdot  (M-j)\cdot  w(\OPT(\I')) \leq \left(\delta e^{200+20M/\log(K)}\right)\cdot  (M-j)\cdot  w(\OPT(\I'))
\end{equation}
on the total charged cost incurred by $\alg$ (in instance $\I$) because of pairs in $B_\calP \cap D$. To finish the proof, we sum these upper bounds over all $B\in \calB^{(j)}$ and all $j$ to obtain indeed the second term of the induction hypothesis.
\end{proof}

Given Lemma \ref{lem:main_induction}, it is now straightforward to prove Theorem \ref{thm:main2}. Lemma \ref{lem:main_induction} applied to the main canonical instance $\I$ states that $\alg$ pays at most 
\begin{align*}
     &(880 e^{5})\cdot  \left( \sum_{j=1}^M \log(k_j)\cdot w\left(\OPT(\I)\cap \calB^{(j)}\right)\right) + e^{200+20M/\log(K)}\cdot \left(\sum_{j=1}^{M} \delta \cdot (M-j)\cdot w\left(\OPT(\I)\cap \calB^{(j)}\right)\right)\\
     &\leq O\left(\sum_{j=1}^M \log(k_j)\cdot w\left(\OPT(\I)\cap \calB^{(j)}\right) +\sum_{j=1}^{M} \delta \cdot (M-j)\cdot w\left(\OPT(\I)\cap \calB^{(j)}\right)\right)
\end{align*}
For the first inequality, we use that $M\leq \log(K)$. We obtain that the first term is at most 
$$
    \log(k)\cdot \sum_{j=1}^M w\left(\OPT(\I)\cap \calB^{(j)}\right) \leq \log(k) \cdot w(\OPT(\I))
$$ since all balls in $\calB$ are pairwise disjoint. The second term, is at most 
$$
    \sum_{j=1}^{M} \delta \cdot (M-j)\cdot w\left(\OPT(\I)\cap \calB^{(j)}\right) \leq (\delta M) \cdot w(\OPT(\I))\leq \log(k) \cdot w(\OPT(\I)),
$$ since we have that $M\leq \log(k)/\delta$ (choosing the upper bound $K=k$ is a valid choice). Hence, if $\I$ is an $(\alpha,\delta)$-canonical instance, we indeed have that the cost paid by $\alg$ on this instance is at most $O(\log(k))\cdot w(\OPT(\I))$, which proves Theorem \ref{thm:main2}, and Theorem \ref{thm:main} if we combine it with Lemma \ref{lem:canonical_transform}.

%% file: Applications.tex
The aim of this section is to establish Theorem \ref{thm:tree} and Theorem \ref{thm:decreasing}. To this end we prove Lemma \ref{lem:subdivision_works_greedy_3}. We recall that in this section, all results apply only to $\mathrm{Greedy}_3$ so $\alg$ will be a shorthand for $\mathrm{Greedy}_3$ in this whole section, unlike in Section \ref{sec:main} where $\alg$ would mean that we could use any of the contraction rules.

\begin{lemma}
\label{lem:subdivision_works_greedy_3}
Suppose we are given an instance $\I$ of Steiner Forest, with $k$ pairs of terminals. Then we can construct an instance $\I'$ of Steiner Forest that satisfies the following, where $\alg = \mathrm{Greedy}_3$:

 \begin{enumerate}
     \item[(a)] $\cost_\alg(\I) = \cost_\alg(\I')$.
     \item[(b)] The contraction of all pairs in instance $\I'$ when running $\alg$ is exactly equal to 1.
     \item[(c)] $\I'$ has $k'= O(k^2)$ terminal pairs.
     \item[(d)] The set of terminals (vertices that appear in a pair) is the same for instance $\I$ and $\I'$.
 \end{enumerate}
\end{lemma}

\begin{proof}
The idea is to subdivide each pair $p_i$ in instance $\I$ into $O(k)$ pairs that are made of pairs of previously arrived terminals that lie on the path $P$ that is used to connect $p_i$. It will be clear from construction that Property (d) holds. Formally, we will construct an instance $\I'$ of Steiner Forest as follows. Beginning with the empty set of terminal pairs, for each pair $\{s_i, t_i\}$ that arrives in $\I$ we will construct a list of pairs $\mathcal{P}(\{s_i,t_i\})$ and add them to our instance $\I'$. The order in which we add pairs to $\I'$ determines the arrival order. We next describe how to construct $\mathcal{P}(\{s_i,t_i\})$. Suppose that when $\{s_i, t_i\}$ appears in $\I$ the algorithm $\alg$ connects the pair using the path $P=s_i, v_{1}, \dots, v_{\ell}, t_i$. Let $s_i, v_{1}', \dots, v_{\ell'}', t_i$ be the sub-sequence of $P$ specified by the contraction rule 3. Recall that this sub-sequence contains only $s_i,t_i$ and previously arrived terminals. We add to $\mathcal{P}(\{s_i,t_i\})$ the subset  of $\{\{s_i, v_{1}'\}, \{v_{1}', v_{2}'\},\ldots,  \{v_{\ell'-1}', v_{\ell'}'\}, \{v_{\ell'}', t_i\}\}$ containing all the pairs such that their distance in the contracted metric is not yet 0 (here we consider the contracted metric just before connecting the pair $p=\{s_i,t_i\}$ in instance $\I$). For any pair $\{s',t'\}$ in $\mathcal{P}(\{s_i,t_i\})$ we say that $\{s_i,t_i\}$ is the parent of $\{s',t'\}$.

Now that we have constructed $\I'$, we argue that it satisfies the properties of the lemma. To show that $\cost_\alg(\I) = \cost_\alg(\I')$ we will prove the following stronger statement. For any pair $p_i=\{s_i,t_i\}$,
\begin{enumerate}
    \item $\alg$ spends exactly the same cost for the pairs in $\mathcal{P}(\{s_i,t_i\})$ in instance $\I'$ that it was spending for the pair $p_i=\{s_i,t_i\}$ in instance $\I$, i.e.
    \begin{equation*}
        \cost_\alg(\I,p_i)=\cost_\alg(\I',\calP(\{s_i,t_i\})).
    \end{equation*}
    \item The contracted metric $G^{(i)}$ after revealing all the pairs up to pair $p_i$ in instance $\I$ is exactly the same as the contracted metric $G'^{(\tau)}$ after revealing all the pairs $\bigcup_{i'\leq i}\calP(\{s_{i'},t_{i'}\})$ in instance $\I'$.
\end{enumerate}

We will prove this statement by induction on $i$. It is vacuously true for $i=0$ since no pair appeared so far in both instances. Assume this is true up to pair $i-1$, and let us prove the statement for step $i$. Let $v'_0, v'_1, \dots, v'_{\ell'}, v'_{\ell'+1}$ be the sub-sequence of previously arrived terminals in the path $P$ that $A$ uses to connect $s_i$ and $t_i$ in instance $\I$, where $v'_0=s_i$ and $v'_{\ell'+1}=t_i$.
Suppose that for some $0\leq j \leq \ell'$ the pair $\{v'_j, v'_{j+1} \}$ is not in $\mathcal{P}(\{s_i,t_i\})$. Then by construction we know that $v'_j$ and $v'_{j+1}$ were already at distance $0$ in the contracted metric before the arrival of $\{s_i,t_i\}$ during the execution of $\alg$ on $\I$, i.e.
\begin{equation*}
    d_{G^{(i-1)}}(v'_j, v'_{j+1})=0.
\end{equation*}
As a result the portion of the path $P$ between $v'_j$ and $v'_{j+1}$ contributes $0$ to the value $\cost_\alg(\I,p_i)$ hence no cost is lost. Now suppose that the pair $\{v'_j, v'_{j+1}\}$ is in $\mathcal{P}(\{s_i,t_i\})$, and that when presented with the pair $\{v'_j, v'_{j+1}\}$ in instance $\I'$, $\alg$ uses a path $P'$ to connect the pair. Denote by $P_{[v'_j,v'_{j+1}]}$ the restriction of the path $P$ to vertices that appear in-between vertices $v'_j, v'_{j+1}$ (included). We claim that $P'=P_{[v'_j,v'_{j+1}]}$. To see this, let us consider the first time this is not the case, this would mean that the path $P'$ costs less to $\alg$ in instance $\I'$ than what $P_{[v'_j,v'_{j+1}]}$ costed to $\alg$ in instance $\I$. But since we assumed that the contracted metrics were identical up to pair $p_{i-1}$, $\alg$ could have replaced the path $P_{[v'_j,v'_{j+1}]}$ by $P'$ to pay less. This is a contradiction to the fact that $\alg$ always takes the shortest path. Hence the total cost for the pair $p_i$ is preserved and we have that 
\begin{equation*}
    \cost_\alg(\I,p_i)=\cost_\alg(\I',\calP(\{s_i,t_i\})).
\end{equation*}
To see why we have the second property, now note that we have that $P'=P_{[v'_j,v'_{j+1}]}$, in particular on the path $P'$ there is no previously arrived terminal other than $v'_j,v'_{j+1}$. By contraction rule 3, it means that $\alg$ simply adds an edge $\{v'_j,v'_{j+1}\}$ of weight $0$ in the contracted metric. But by the definition of contraction rule 3, this edge was also added by $\alg$ when connecting the pair $p_i$ in instance $\I$. Reciprocally, it is clear that if a shortcut was added by $\alg$ when connecting the pair $p_i$ in instance $\I$, it must be between two previously arrived terminals $v'_{j'},v'_{j'+1}$ that are consecutive on the path $P$. If these two previously arrived terminals were already at distance $0$ in the contracted metric, then the shortcut does not change the metric. Otherwise if $v'_{j'}$ and $v'_{j'+1}$ were not at distance $0$ then we have that $\calP(\{s_i,t_i\})$ also contains the pair $\{v'_{j'},v'_{j'+1}\}$ hence this shortcut will also be added in the metric when running instance $\I'$. This proves that the contracted metrics are indeed the same after step $i$. We also have Property (b) since, by construction, there is no previously arrived terminal on the path $P'$. Hence it must be that $\alg$ pays the cost of the path $P'$ in the original metric (the only way to pay less than the length in the original metric is to use a shortcut that was added before, but such a shortcut must connect previously arrived terminals).

Lastly, we argue that the number of terminals $k'$ in $\I'$ is $O(k^2)$. Since there are $k$ pairs in instance $\I$, we know that there are at most $2k$ terminals in total. Thus there are at most $2k$ terminals along any path $\alg$ takes to connect a terminal pair in $\I$ (going twice through the same terminal cannot happen since $\alg$ takes the shortest path). Summing this bound over all terminal pairs in $\I$, we get that $\I'$ has at most $2k^2$ terminals.
\end{proof}

With Lemma \ref{lem:subdivision_works_greedy_3}, we are able to prove Theorem \ref{thm:tree}.

\begin{proof}[Proof of Theorem \ref{thm:tree}]
Construct $\I'$ from $\I$ as in Lemma \ref{lem:subdivision_works_greedy_3}. Then we can write 
\begin{align*}
    \cost_\alg(\I) = \cost_\alg(\I') & \leq  O(\log(k')\cdot \log\log(k'))\cdot w(\OPT(\I')) \\
    & \leq O(\log(k')\cdot \log\log(k'))\cdot w(T^\star(\I)) \\
    & \leq  O(\log(k)\cdot \log\log(k))\cdot w(T^\star(\I)),
\end{align*}
where the first equality follows by Property (a) of Lemma \ref{lem:subdivision_works_greedy_3}, the second inequality from Theorem \ref{thm:main} applied to instance $\I'$ (using that contraction of all pairs is 1 by Property (b) of Lemma \ref{lem:subdivision_works_greedy_3}), the third inequality from the fact that the optimum tree solution to instance $\I$ denoted $T^\star(\I)$ is also a feasible solution to instance $\I'$ (because by Property (d) of Lemma \ref{lem:subdivision_works_greedy_3} the set of terminals does not change) and the last inequality from the fact that $k'=O(k^2)$ by Property (c) of Lemma \ref{lem:subdivision_works_greedy_3}.
\end{proof}

Next we prove Theorem \ref{thm:decreasing}. For this, we require some additional notation. Let $\I$ be an instance of Steiner Forest and $\calF$ a feasible solution consisting of trees $T_1, \dots, T_{m}$. Then the width of a tree $T_j$ with respect to an instance $\I$, denoted $\wi(T_j,\I)$, is defined to be the largest distance between any terminal pair in $T_j$ in the original metric, i.e.
\begin{equation*}
    \wi(T_j,\I) = \max_{u\in \calT \cap T_j} d_{G}(u,\overline{u}).
\end{equation*} where $\overline{u}$ is the terminal that $u$ should be connected to in instance $\I$, and $\calT$ is the set of all terminals in instance $\I$. This is exactly the same definition of width as in \cite{gupta2015greedy}. For such a forest $\calF=(T_1,\ldots, T_q)$ we define the potential 
$$\Phi(\calF,\I) = w(\calF) + \sum_{j=1}^{q}\wi(T_j,\I).$$ 

Where $w(\calF)$ is the cost of the forest. Note that $w(\calF)\leq \Phi(\calF,\I) \leq 2w(\calF)$ for any forest $\calF$ that is a feasible solution to instance $\I$. We will use this property to prove Theorem \ref{thm:decreasing}.

\begin{proof}[Proof of Theorem \ref{thm:decreasing}]
We construct the instance $\I'$ from instance $\I$ exactly as how we transformed the instance in Lemma \ref{lem:subdivision_works_greedy_3}. We show that if the costs of shortest paths are non-increasing over time, it must be that
\begin{equation*}
    w(\OPT(\I')) = O(w(\OPT(\I)).
\end{equation*}
Once this is established, the result follows as we have

\begin{align*}
    \cost_\alg(\I) = \cost_\alg(\I') & = O(\log(k')\cdot \log\log(k'))\cdot w(\OPT(\I')) \\
    & = O(\log(k)\cdot \log\log(k))\cdot w(\OPT(\I)) \\
\end{align*}
again by Lemma \ref{lem:subdivision_works_greedy_3} and Theorem \ref{thm:main} (recall that subdividing pairs to create instance $\I'$ as in Lemma \ref{lem:subdivision_works_greedy_3} guarantees that the contraction is 1 for all pairs in the instance $\I'$).

We argue that $w(\OPT(\I')) = O(w(\OPT(\I))$, using essentially the same potential function argument than that of \cite{gupta2015greedy}. The idea is to begin with the solution $\OPT(\I)$ and add additional connections to produce a feasible solution $\calF$ to $\I'$, where $w(\calF) \leq 2 w(\OPT(\I))$. Since $w(\OPT(\I')) \leq w(\calF)$ if we succeed the proof is complete.

As previously stated, we initialize $\calF$ to be $\OPT(\I)$. We will only add edges to $\calF$ hence it will be clear that $\calF$ will always be a feasible solution to instance $\I$. Now for each terminal pair $\{s_i, t_i\}$ that arrives in $\I$ we construct $\calF'$ from $\calF$ as follows. Suppose that all the pairs in $\mathcal{P}(\{s_i,t_i\})$ are already connected by the current solution $\calF$, then nothing needs to be done.

Suppose otherwise that the solution $\calF$ does not connect all pairs in $\mathcal{P}(\{s_i,t_i\})$. By re-indexing let $T_1, \dots, T_{q}$ be the components of $\calF$ that contain the terminals that appear in the pairs of $\mathcal{P}(\{s_i,t_i\})$, ordered so that $$\wi(T_1) \geq \dots \geq \wi(T_q).$$ 

Construct $\calF'$ from $\calF$ by adding the edges that $\alg$ uses to connect the pairs in $\mathcal{P}(\{s_i,t_i\})$ in instance $\I'$.  We claim that the total cost of these edges (which is exactly equal to $\cost(\I,p_i)$) is at most the width of $T_{q}$ (with respect to instance $\I$), i.e.
\begin{equation*}
    \cost(\I,p_i)\leq \wi(T_q,\I).
\end{equation*}
To see this first note that the tree $T_q$ either contains the pair $\{s_i,t_i\}$ or a pair that appeared even before $\{s_i,t_i\}$ in instance $\I$ (here we use that $\calF$ is a feasible solution to instance $\I$). Hence if the sequence of shortest paths $d_1=d_G(s_1,t_1), \dots ,d_k=d_G(s_k,t_k)$ is non-increasing, then

$$\wi(T_q,\I) \geq  \min_{i'\leq i}d_{i'} = d_i \geq c_i = \cost(\I,p_i).$$

For the same reason if the sequence  of costs $c_1 = \cost(\I,p_1) ,\dots, c_k = \cost(\I,p_k)$ is non-increasing, then $$\wi(T_q,\I) \geq \min_{i'\leq i}d_{i'} \geq  \min_{i'\leq i}c_{i'} = c_{i} = \cost(\I,p_i).$$ We use here that $d_i\geq c_i$ for all $i$. This proves our claim.

We now note that, if we needed to add edges because the current solution $\calF$ did not connect all the pairs in $\mathcal{P}(\{s_i,t_i\})$, then we obtained a new solution $\calF'$ such that
\begin{equation*}
    w(\calF')\leq w(\calF) + \cost(\I,p_i),
\end{equation*}
and at least two of the components in $T_1,T_2,\ldots ,T_q$ were merged into one component. Hence
\begin{equation*}
    \wi(\calF',\I)\leq \wi(\calF,\I)-\min_{1\leq i\leq q}\wi(T_i,\I)\leq \wi(\calF,\I)-\wi(T_q,\I).
\end{equation*}
In particular we obtain
$$\Phi(\calF',\I)-\Phi(\calF,\I) \leq  \cost(\I,p_i) -\wi(T_q,\I) \leq 0,$$ by the above remarks. We then update $\calF$ to be $\calF'$.

After completing the process we have that $\calF$ connects all pairs in $\I'$, as required and that the potential function did not increase. Therefore we have that
\begin{equation*}
    w(\calF')\leq \Phi(\calF',\I)\leq \Phi(\calF,\I) \leq 2w(\calF),
\end{equation*}
as desired.
\end{proof}

%% file: OpenProblems.tex
In this last section, we discuss problems that are left open by our analysis. The obvious open problem is to prove that greedy is $O(\log(k)\cdot \log(k)\log(k))$-competitive on general instances. The immediate idea to prove the latter would be to improve Theorem \ref{thm:main} and show that the assumption on the contraction is not needed. However, this assumption is used crucially in one place of our proof. That is, during the inductive proof of Lemma \ref{lem:main_induction} where we argue that we can ``cluster'' pairs of terminals inside a big dual ball and pretend that we could run a greedy algorithm on them independently of the rest of the instance. In this argument we crucially rely on the fact that these pairs have \textit{both} endpoints inside the ball. This has to be the case if the contraction of pairs is small (because to leave the ball they have to cross the border $\partial B$ which induces a path much longer than the cost of the pairs at hand). However if the contraction is very big, it might be that the other endpoint is very far from the ball $B$. In this case, the inductive step as it is now does not work. Lemma \ref{lem:main_induction} is the only place where the assumption on contraction is really needed, we believe the rest of Section \ref{sec:main} works without this assumption. Hence circumventing this issue without breaking anything else in the proof would immediately prove that greedy is $O(\log(k)\cdot \log\log(k))$-competitive (this would apply to any of the three contraction rules).

Another way around this issue would be to use Theorem \ref{thm:main} as a black box and argue that it implies an improved bound in general. This seems to be a promising approach and we conjecture that pairs with contraction bigger than $\textrm{polylog}(k)$ should not be an issue. In fact, we are not even aware of an instance for which the cost of pairs with contraction \textit{greater} than $\textrm{polylog}(k)$ is not within constant factor of the total cost of pairs with contraction \textit{smaller} than $\textrm{polylog}(k)$ and such that greedy is $\omega(1)$-competitive on this instance! Unfortunately, such an argument has remained elusive to us. Another direction would be to use Theorem \ref{thm:tree} along with some more involved potential function arguments that show that worst-case instances are in fact single tree instances. Interestingly, it is easy to show that for $\textrm{Greedy}_2$ the worst-case instance is indeed one where the optimum solution is a single tree, but we are not aware of a way to prove Theorem \ref{thm:tree} for this version of greedy.

Another intriguing question is that, even circumventing the issue above, it seems our techniques will reach their limits at the ratio $O(\log(k) \cdot \log\log(k))$. But the exact conjecture of Awerbuch, Azar, and Bartal is that greedy is $\Theta(\log(k))$-competitive. Hence one would need to get rid of the $\log\log(k)$ in our proof but this seems challenging. This additionnal factor is coming from the proof of Awerbuch, Azar, and Bartal in which they scale down the dual balls by a factor of roughly $\log(k)$ to be able to apply some girth argument. Hence to cluster the pairs inside these balls, we need to look essentially $\log\log(k)$ cost classes further and we can only get a $1/\log\log(k)$ fraction of the total cost by doing this. Hence it seems that scaling down by $\log(k)$ is a bad idea in our case. One would like to scale down only by a constant factor but the problematic examples of \cite{SODA96,chen2010designing} show that it might be that one has to lose all but a $1/\log(k)$ fraction of the dual balls. Such an argument seems very unclear to us if the girth argument is broken, which would be the case if we scale down the dual balls by only a constant factor.

%% file: appendix.tex
\subsection{A lower bound of $\Omega(\log(k))$ with a single cost class and such that all pairs have contraction equal to 1}
\label{subsec:appendixexample}
We present here the example given by \cite{chen2010designing,gupta2015greedy} which shows that the analysis of Steiner Forest for a single cost class is exactly tight. That is, it might be that the cost incurred by greedy on one cost class is already $\Omega(\log(k))\cdot w(\OPT(\I))$. Consider an unweighted cubic graph with $n$ vertices and girth $g=c\log(n)$ for some constant $c$ (see \cite{biggs1998constructions} for a construction of such graphs). Fix a spanning $T$ of $G$ and let $E'=E\setminus E(T)$ be the non-tree edges. Set the length of edges in $E(T)$ to $1$ and the length of edges in $E'$ to $g/2$.

Finally, consider $M$ a maximum matching in $G'=(V,E')$. The demand set $\calP$ will be simply the matching $M$, and the metric the weighted graph as defined above. By induction on the number of pairs already arrived, we show that greedy pays exactly $g/2$ for all pairs in $M$ (regardless of the contraction rule chosen for greedy). Assume this is true so far, that is we revealed a subset $M'\subseteq M$ of the matching to greedy, and greedy bought exactly those edges in $M'$. Now reveal the next edge $e=\{s,t\}$ in the matching. To connect $s$ to $t$, greedy can either buy the edge $\{s,t\}$ which costs $g/2$, or try some other path. However, because the girth in $G$ is at least $g$, this other path must contain at least $g-1$ edges. Denote by $m$ the number of edges on this path that do not belong the tree $T$ and by $m'$ the number of these edges on the path that were already bought by greedy. Note that $m'\leq \left\lceil m/2\right\rceil$ since $M'$ is a matching. Then we see that the cost of this other path must be at least
\begin{equation*}
    (g-1-m)+(m-m')\cdot (g/2)\geq  (g-1-m)+\left\lfloor m/2\right\rfloor\cdot g/2\geq g-2 > g/2.
\end{equation*}
Hence the unique shortest path is again to buy the single edge $\{s,t\}$ which proves the induction. Note that this also proves that all pairs in the instance have contraction exactly 1.

We can now lower bound the cost paid by greedy. Since the maximum degree in the graph is $3$, it must be that $|M|=\Omega (n)$. This means that greedy pays a cost $\Omega (n\log(n))$ because of pairs in $M$ (which all cost \textit{exactly} the same by construction). However, a better solution would be simply to buy the spanning tree of $G$ that costs $n-1$ by construction. Hence we get a competitive ratio of $\Omega(\log(n))$ for greedy (with any of our 3 contraction rules), such that there is only one cost class, and all pairs have contraction equal to 1.

What is intuitively happening in this example is that greedy is tricked into thinking that all the pairs should be in their own component, instead of buying one big spanning tree for all the terminals. It also highlights that greedy for Steiner Tree and greedy for Steiner Forest are indeed very different algorithms on instance where the optimum is a tree, even though they share the same name of ``greedy''.

\subsection{Proof of Lemma \ref{lem:AAB}}

Here we prove Lemma \ref{lem:AAB}. It will be helpful to recall Moore's bound. Recall that the girth of a graph is the length of the shortest cycle.

\begin{theorem}[Moore's bound, see \cite{bollobas2004extremal},\cite{alon2002moore}]
\label{thm:moore}
Every graph with at least $2n^{1+\frac{1}{p}}$ edges has girth at most $2p$. 
\end{theorem}

\begin{proof}[Proof of Lemma \ref{lem:AAB}]

Recall that we are interested in a subset of pairs $\calP'$ that all belong to a same cost class $\calP'$. That is, all the pairs in $\calP'$ cost the same value $c$. We will place dual balls as desired in the statement of Lemma \ref{lem:AAB}. To do this we will try to place a dual ball around the endpoint of a pair $p$ as it arrives. If this is not possible without intersecting previously placed balls, we then skip the pair $p$. It will be clear by construction that the dual solution will be feasible. As we construct the feasible dual solution $\calB$, we will maintain an auxiliary graph $G'=(V',E')$ which will be unweighted. The quantity of interest will be the number of edges in $E'$. The vertex set $V'$ will correspond to a subset of terminals that appear in the pairs in $\calP'$. 

We proceed as follows. When a pair $p=\{s,t\}$ belonging to  $\calP'$ arrives, we try to place a ball of radius
\begin{equation*}
    r=\frac{c}{8\log(|\calP'|)}
\end{equation*}
around either one of $s$ or $t$. If one of these two balls can be placed without intersecting balls previously placed in $\calB$ we place it. We also add its center (either $s$ or $t$) to the vertex set $V'$ of the auxiliary graph $G'$. Otherwise, if none of these balls can be placed without intersecting balls already placed in $\calB$ we do not add any ball. However, we identify one ball $B\in \calB$ that intersects with the ball of radius $r$ around $s$ and another ball $B'\in \calB$ for $t$. Let $s'$ and $t'$ be their centers. By construction we have that $s',t'\in V'$. We can add the edge $\{s',t'\}$ in the auxiliary graph. We have the following claim.

\begin{claim}
The girth of $G'$ is at least $2\log(|\calP'|)$.
\end{claim}
\begin{proof}[Proof of Claim]
Suppose that the claim does not hold, and consider a cycle of length $\ell<2\log(|\calP'|)$ in $G'$ as in Figure \ref{fig:girth}.
\begin{figure}
    \centering
    \includegraphics[scale=0.8]{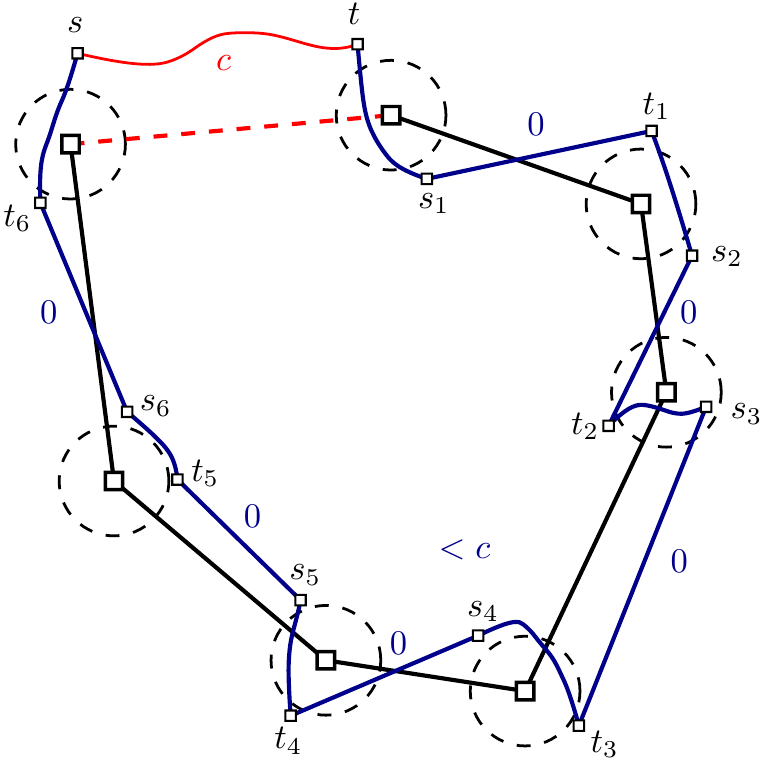}
    \caption{The girth argument.}
    \label{fig:girth}
\end{figure}
Consider the last edge that was added to the cycle. Suppose it was added because of terminal pair $\{s,t\}$ (in red on Figure \ref{fig:girth}). Each other edge in the cycle must have been created by a previous pair $\{s_m,t_m\}$ for $1\leq m < \ell$. By triangle inequality we have, for all $m< \ell$ that
\begin{equation*}
    d_G(t_m,s_{m+1})\leq d_G(t_m,u)+d_G(u,s_{m+1})\leq 4r,
\end{equation*}
where $u$ is the center of the ball that was close to both $t_m$ and $s_{m+1}$. Similarly, $d_G(t,s_1)\leq 4r$ and $d_G(t_{m-1},s)\leq 4r$. Since all pairs $\{s_m,t_m\}_{m\leq \ell-1}$ arrived before $\{s,t\}$, it must be that the shortest path in the current metric $G'$ (when the pair $\{s,t\}$ arrives) satisfies 
\begin{equation*}
    d_{G'}(s_m,t_m)=0
\end{equation*} (this holds for all three contraction rules). Hence, by going through the terminals $t,s_1,t_1,s_2,t_2,\ldots ,s_{\ell-1},t_{\ell-1},s$, greedy could have paid at most 
\begin{equation*}
    (4 r)\cdot \ell < (8 r)\cdot  \log(|\calP'|) \leq c,
\end{equation*}
which is a contradiction on the fact that greedy should always take the shortest path available. Hence the girth in $G'$ is at least $2\log(|\calP'|)$. We note that this whole argument also holds if we are in the case of Steiner Forest in decreasing metrics where additional edges are revealed over time. Indeed, these additional edges cannot make a shortest path longer. 
\end{proof}

Applying Moore's bound to the graph $G'$ (which has at most $|\calP'|$ vertices by construction) gives that $|E'| < 4 |V'|\leq 4|\calP'|$. Since the cardinality of $\calP'$ equals the number of vertices $|V'|$ in $G'$ plus the number of edges $|E'|$ in $G'$, and $|V'|=|\calB|$ property (b) of Lemma \ref{lem:AAB} follows. Properties (a), (c) and (d) follow by construction.
\end{proof}